\documentclass[draftcls,onecolumn]{IEEEtran}
\usepackage[utf8]{inputenc}
\usepackage{amsmath,amsfonts}
\usepackage{soul}
\usepackage{verbatim}
\usepackage{graphicx}
\usepackage{float}
\usepackage{amsthm}
\usepackage{subcaption}
\usepackage{paralist}
\usepackage{url}
\usepackage{tikz}
\usepackage{cite}
\newtheorem{Def}{Definition}
\newtheorem{theorem}{Theorem} 
\newtheorem{lemma}[theorem]{Lemma}
\newtheorem{corollary}[theorem]{Corollary} 

\graphicspath{{fig}} 

\newcommand{\myh}{h}
\newcommand{\mya}{a}

\newcommand{\dg}[1]{{\color{red}#1}}

\newcommand{\sj}[1]{{\color{red}#1}}

\usetikzlibrary {shapes.symbols}
\usetikzlibrary{positioning,tikzmark}
\usetikzlibrary{overlay-beamer-styles}
\usetikzlibrary{shapes.geometric} 
\usetikzlibrary{shadings}
\makeatletter
\pgfdeclareshape{hornedsquare}{
    \savedanchor\centerpoint{
        \pgf@x=0pt
        \pgf@y=0pt
    }
    \inheritsavedanchors[from=rectangle]  
    \inheritanchorborder[from=rectangle]
    \inheritanchor[from=rectangle]{center}
    \inheritanchor[from=rectangle]{north}
    \inheritanchor[from=rectangle]{south}
    \inheritanchor[from=rectangle]{east}
    \inheritanchor[from=rectangle]{west}
    \inheritanchor[from=rectangle]{north east}
    \inheritanchor[from=rectangle]{north west}
    \inheritanchor[from=rectangle]{south east}
    \inheritanchor[from=rectangle]{south west}

    \backgroundpath{
        \pgf@xa=-0.15cm 
        \pgf@xb=0.35cm  
        \pgf@ya=-0.15cm 
        \pgf@yb=0.35cm  

        \pgfpathmoveto{\pgfpoint{\pgf@xa+0cm}{\pgf@yb+0cm}}
        \pgfpathlineto{\pgfpoint{\pgf@xb}{\pgf@yb+0cm}}
        \pgfpathlineto{\pgfpoint{\pgf@xb+0cm}{\pgf@ya+0cm}}
        \pgfpathlineto{\pgfpoint{\pgf@xa+0cm}{\pgf@ya+0cm}}
        \pgfpathclose

        \pgfpathmoveto{\pgfpoint{\pgf@xa+0.031cm}{\pgf@yb+0.125cm}}
        \pgfpathlineto{\pgfpoint{\pgf@xa}{\pgf@yb+0cm}}
        \pgfpathlineto{\pgfpoint{\pgf@xa+0.125cm}{\pgf@yb+0cm}}
        \pgfpathclose
        
        \pgfpathmoveto{\pgfpoint{\pgf@xb-0.031cm}{\pgf@yb+0.125cm}}
        \pgfpathlineto{\pgfpoint{\pgf@xb+0cm}{\pgf@yb+0cm}}
        \pgfpathlineto{\pgfpoint{\pgf@xb-0.125cm}{\pgf@yb+0cm}}
        \pgfpathclose
    }
}
\pgfdeclareshape{hornedsquareempty}{
    \savedanchor\centerpoint{
        \pgf@x=0pt
        \pgf@y=0pt
    }
    \inheritsavedanchors[from=rectangle]  
    \inheritanchorborder[from=rectangle]
    \inheritanchor[from=rectangle]{center}
    \inheritanchor[from=rectangle]{north}
    \inheritanchor[from=rectangle]{south}
    \inheritanchor[from=rectangle]{east}
    \inheritanchor[from=rectangle]{west}
    \inheritanchor[from=rectangle]{north east}
    \inheritanchor[from=rectangle]{north west}
    \inheritanchor[from=rectangle]{south east}
    \inheritanchor[from=rectangle]{south west}

    \backgroundpath{
        \pgf@xa=-0.25cm 
        \pgf@xb=0.25cm  
        \pgf@ya=-0.25cm 
        \pgf@yb=0.25cm  

        \pgfpathmoveto{\pgfpoint{\pgf@xa+0cm}{\pgf@yb+0cm}}
        \pgfpathlineto{\pgfpoint{\pgf@xb}{\pgf@yb+0cm}}
        \pgfpathlineto{\pgfpoint{\pgf@xb+0cm}{\pgf@ya+0cm}}
        \pgfpathlineto{\pgfpoint{\pgf@xa+0cm}{\pgf@ya+0cm}}
        \pgfpathclose

        \pgfpathmoveto{\pgfpoint{\pgf@xa+0.031cm}{\pgf@yb+0.125cm}}
        \pgfpathlineto{\pgfpoint{\pgf@xa}{\pgf@yb+0cm}}
        \pgfpathlineto{\pgfpoint{\pgf@xa+0.125cm}{\pgf@yb+0cm}}
        \pgfpathclose
        
        \pgfpathmoveto{\pgfpoint{\pgf@xb-0.031cm}{\pgf@yb+0.125cm}}
        \pgfpathlineto{\pgfpoint{\pgf@xb+0cm}{\pgf@yb+0cm}}
        \pgfpathlineto{\pgfpoint{\pgf@xb-0.125cm}{\pgf@yb+0cm}}
        \pgfpathclose
    }
}
\makeatother

\tikzset{
  adversary/.style = {shape=house, minimum width=0.5cm, minimum height=0.5cm, fill=pink,
                     draw, align=center},
  honest/.style     = {shape=rectangle, minimum width=0.5cm, minimum height=0.5cm, fill=green, draw, align=center},
root/.style     = {honest},
	blkhonest/.style     = {shape=rectangle, minimum width=0.5cm, minimum height=0.5cm,fill=blue!0, draw, align=center},
 	blkad/.style     = {shape=hornedsquare, minimum width=0.5cm, minimum height=0.5cm, fill=red!0, draw, align=center},
     blkadempty/.style     = {shape=hornedsquareempty, minimum width=0.5cm, minimum height=0.5cm, fill=red!0, draw, align=center},
	rec/.style     = {shape=rectangle, minimum width=0.5cm, minimum height=0.5cm, fill=blue!0, draw, align=center},
 hidhonest/.style = {shape=rectangle, minimum width=0.5cm, minimum height=0.5cm, fill=black!15, 
  draw, align=center},
  hidad/.style = {shape=hornedsquare, minimum width=0.5cm, minimum height=0.5cm, fill=black!15,
    draw, align=center},
  hidadempty/.style = {shape=hornedsquareempty, minimum width=0.5cm, minimum height=0.5cm, fill=black!15,
    draw, align=center}
}
\tikzstyle{block} = [draw,rectangle] 

\title{Security, Latency, and Throughput of Proof-of-Work Nakamoto Consensus}

\author{
Shu-Jie Cao and Dongning Guo\thanks{The authors are with the Department of Electrical and Computer Engineering, Northwestern University, Evanston, Illinois.  Email: \{shujie.cao, dguo\}@northwestern.edu.}}

\begin{document}

\maketitle

\begin{abstract}
    This paper investigates 
    the fundamental trade-offs between 
    block safety, confirmation latency, and transaction throughput of proof-of-work (PoW) longest-chain fork-choice protocols, also known as PoW Nakamoto consensus. New upper and lower bounds are derived for the probability of block safety violations as a function of honest and adversarial mining rates, 
    a block propagation delay limit, and confirmation latency measured in both time and block depth. The results include the first non-trivial closed-form finite-latency bound applicable across all delays and mining rates up to the ultimate fault tolerance. Notably, the gap between these upper and lower bounds is narrower than previously established bounds for a wide range of parameters relevant to Bitcoin and its derivatives, including 
    Litecoin and Dogecoin, as well as 
    Ethereum Classic. Additionally, the 
    study uncovers a fundamental trade-off between transaction throughput and confirmation latency, ultimately determined by the desired fault tolerance and the rate at which block propagation delay increases with block size.
\end{abstract}

\thispagestyle{plain} %
\pagestyle{plain} 

\section{Introduction}
\label{s:intro}

Nakamoto consensus is a foundational protocol in blockchain technology, widely adopted across cryptocurrencies to enable decentralized agreement on the blockchain's state. In these systems, participants organize incoming transactions into blocks, with each block pointing to a unique preceding block, thereby forming a linked chain of blocks. Bitcoin and some of its derivatives use a proof-of-work (PoW) Nakamoto consensus mechanism, where miners continually compete to solve complex cryptographic puzzles; solving a puzzle grants a miner the right to add a new block to the system. The reader is referred to~\cite{antonopoulos2017mastering, narayanan2016bitcoin} for detailed descriptions of the consensus mechanism.

A fork occurs when two or more blocks point to the same parent block. Forks can arise from adversarial mining or naturally due to block propagation delays, leading honest miners to have differing views of the blockchain state. Nakamoto consensus is built around the longest-chain fork-choice rule, which mandates that an honest miner: 1) always mines at the tip of a longest chain in its view, and 2) immediately shares any new chain extensions with other miners through a peer-to-peer network.

The most widely adopted rule is to commit a block once it is confirmed by a sufficient number of subsequent blocks. A rule of thumb for Bitcoin is six block confirmations. As an alternative, a block on a longest chain may be committed if a sufficient amount of time has lapsed since the block was seen. In either scenario, the safety of a committed block could still be violated if another chain, which excludes that block, is to emerge as a longest chain at any subsequent time. Since such a violation presents a {\em double-spending} opportunity with potentially grave consequences, it is crucial to make the probability of violation very small under the adversary's strongest feasible attacks.

In 2020, Dembo et al.~\cite{dembo2020everything} and Ga\v{z}i et al.~\cite{gazi2020tight} used very different methods to determine the ultimate fault tolerance for the Nakamoto consensus. As an abstraction of the peer-to-peer network, suppose once a block is known to one honest miner, it is known to all honest miners within $\Delta$ seconds. Suppose also the adversary mines $a$ blocks per second and the honest miners collectively mine $h$ blocks per second, all through proof of work (PoW). Evidently, the corresponding inter-arrival times are $1/a$ and $1/h$ seconds, respectively. Then as long as
\begin{align} \label{eq:a<}          
    \frac1a
    >
    \Delta +  \frac1h,
\end{align}
every block is safe with arbitrarily high probability after sufficient block confirmations, regardless of the adversary's attack strategy; conversely, if $1/a<\Delta+1/h$,
no 
block is ever safe in the sense that a well-known attack can violate its safety with probability 1, regardless of the number of block confirmations.

If~\eqref{eq:a<} answers the ``capacity'' question, then a ``finite-blocklength'' question ensues: Suppose the parameters satisfy~\eqref{eq:a<},
how safe is a committed block after $k$ block confirmations? 
Ga\v{z}i, Ren, and Russell~\cite{gazi2022practical} 
provided a safety guarantee 
via numerically evaluating 
a Markov chain, 
which is only practical for 
small confirmation depths 
and parameters that are 
far away from the ultimate fault tolerance. 
By analyzing races between renewal processes, Li, Guo, and Ren
~\cite{li2021close} 
presented a closed-form result focusing on confirmation by time.
Guo and Ren~\cite{guo2022bitcoin} 
introduced a rigged model, where honest blocks mined within $\Delta$ seconds of any previous blocks are flipped to become adversarial blocks. With this additional advantage, the adversary's optimal attack reduces to the well-understood private mining attack.
A simple closed-form latency-security guarantee
is given in~\cite{guo2022bitcoin}, although the rigging technique also weakens the
fault tolerance from~\eqref{eq:a<} to $\beta < 1 - (1/2) e^{\lambda \Delta}$, where $\lambda=a+h$ denotes the total mining rate and $\beta=a/(a+h)$ denotes the fraction of adversarial mining.
Doger and Ulukus~\cite{doger2024refined} 
refined the rigged model 
of~\cite{guo2022bitcoin} 
by 
not flipping honest blocks which see no other honest blocks within the $\Delta$ period before its mining,
thereby improving the security-latency bounds.
The same authors generalized their study based on the rigged model to the cases of different network delay assumptions
in~\cite{doger2024transactioncapacitysecuritylatency}
and~\cite{doger2024powsecuritylatencyrandomdelays}.
However, 
all of the preceding rigged models 
transfer some honest mining rate to the adversary, yielding strictly positive losses in the fault tolerance, i.e.,~\eqref{eq:a<} is not achieved.

In this paper, we carefully define a continuous-time model for block mining processes, blockchains, and safety violation for blocks. We develop a new closed-form latency-security trade-off for the Nakamoto consensus which is the first to achieve the ultimate fault tolerance~\eqref{eq:a<}.
The new results here also substantially reduce the previous gap between upper and lower bounds on the latency-security trade-off in many scenarios of interest.
\begin{table}
\begin{center}
\begin{tabular}{|r||r|r|r|r|r|}
\hline
$\beta$ 
& new lower & new upper
& \cite{guo2022bitcoin} 
& \cite{doger2024refined} 
& \cite{gazi2022practical} 
\\ 
\hline
10\% & 6 &8&7&7&7 \\
\hline
20\% & 13 &16&17&15&15 \\
\hline
30\% & 32 &39&42& 38&N/A \\
\hline
40\% & 122 &156&187& 165&N/A \\
\hline
 49.2\% &26720 & 44024& $\infty$ & $\infty$ &N/A\\
\hline
\end{tabular}    
\end{center}
\caption{
Bounds on confirmation depth for 99.9\% safety.  $
\lambda=1/600$ (blocks/second); $\Delta=10$ seconds.}
\label{tb:depth}
\end{table}
In Table~\ref{tb:depth} we compare our new bounds to the best prior bounds on the minimum confirmation depth that guarantees 99.9\% safety.
The total mining rate is assumed to be Bitcoin's 1 block every 600 seconds.  The propagation delay bound $\Delta$ is set to 10 seconds.  
Our results encompass various adversarial mining fractions from $10\%$ to $49.2\%$.
We include bounds for 10\% and 20\% adversarial mining taken directly from~\cite{gazi2022practical}, which provides no result for higher adversarial mining rates because the numerical method therein is 
infeasible for parameters that are 
close to the 
fault tolerance.  Evidently, the new upper bound is much closer to the lower bound than prior bounds in challenging cases.
Moreover, the new bounds apply to all parameter values satisfying~\eqref{eq:a<}, which in this case is up to $
\beta=49.8\%$.  In contrast, the fault tolerance in~\cite{guo2022bitcoin} is $49.16\%$, hence for $49.2\%$, it is
indicated by $\infty$ in Table~\ref{tb:depth}.

The inverse of the right-hand side of~\eqref{eq:a<} is known to be the minimum growth rate of honest blocks' heights~\cite{pass2017analysis, dembo2020everything, gazi2020tight}.  
If this growth rate exceeds the adversarial mining rate (as in~\eqref{eq:a<}), the chance that more adversarial blocks are mined to overwhelm honest blocks to create a violation is very small for a confirmation latency sufficiently high.
However, even if fewer adversarial blocks are mined than the height growth of honest blocks (i.e., the adversary has a deficit), a violation is still possible when honest miners with different views extend different chains to balance each other's blocks on enough heights (i.e., the deficit is covered by those ``balanced heights''). This paper develops a novel approach to the latency-security problem by examining the number of balanced heights.

Upon each opportunity to have a balanced height,
there is a chance that 
all honest 
miners fail to balance, rather they synchronize to a common highest chain.
Since honest blocks' height growth rate exceeds the adversarial mining rate, there is a good chance that the adversary is left behind permanently, so no alternative highest chain is ever produced to cause a violation of the target block.
Using this insight, we show that the number of balanced heights is stochastically dominated by essentially a geometric random variable.
Interestingly, the aforementioned synchronizing block falls under the category of ``Nakamoto blocks'' defined in~\cite{dembo2020everything}.  Without using the complicated machinery of trees and tines as in~\cite{dembo2020everything, gazi2022practical}, our analysis yields 
tight, easily computable finite-latency security bounds 
which achieve and imply
the ultimate fault tolerance~\eqref{eq:a<}.

The 
tight 
latency-security trade-off enables a user to choose an efficient confirmation latency that aligns with their desired security level.
There has been ample evidence (e.g.,~\cite{fechner2022calibrating}) that the block propagation delay is approximately an affine function of the block size. 
In this paper, we explore the trade-off between the total throughput in kilobytes (KB) per second and the confirmation latency (in minutes), both of which depend on the block size and mining rate.  This paper provides improved guidance for selecting parameters that optimize key performance indicators.

The remainder of this paper is organized as follows:
Sec.~\ref{s:model}, 
introduces a mathematical model for the Nakamoto consensus.
Sec.~\ref{s:sufficiency} explains why the model is sufficient for our purposes. 
Sec.~\ref{s:theorem}
presents the main theorems. Sec.~\ref{s:essential} introduces the most essential elements of the proof, followed by a sketch of the analysis.
Sec.~\ref{s:numerical} presents the numerical results, showcasing the practical implications. 
Detailed proofs are relegated to the appendices.


\section{Model}
\label{s:model}
In this section, we 
put forth a mathematical model for block mining processes and blockchains, enabling a precise definition of key 
concepts such as
confirmation, latency, and safety violations. In Sec.~\ref{s:sufficiency}, we will delve into why the model is sufficient for the latency-security analysis.
\subsection{Block and Chain}
\label{s:block}
We consider a continuous-time model, where 
blocks can be mined at arbitrary times. 
Throughout this paper, ``by time $t$" means ``during $(0,t]$".
\begin{Def}[block] 
    The genesis block, also referred to as block 0, is mined first.
    Subsequent blocks are 
    numbered sequentially in the order they are mined, starting from
    block 1, block 2, and so on. 
    From block 1 onward,
    each block has a unique parent block,  
    which is mined strictly before it.
    The parent of block 1 is the genesis block.
    We use $t_b$ to denote the time when block $b$ is mined.
\end{Def}

\begin{Def}[blockchain and height]
    A sequence of 
    $n+1$ block numbers, $(b_0, b_1,\dots,b_n)$, defines a blockchain (or simply a chain) if $b_0=0$, and for every $i=1,\dots,n$, block $b_{i-1}$ is block $b_i$'s parent.
    As a convention, we say the height of this chain is $n$.
    Since every chain can be uniquely identified by its last block, 
    we often refer to the preceding chain as chain $b_n$. Also, we say block $b_{i+1}$ extends chain $b_i$ and that chain $b_i$ is a prefix of chain $b_n$ for every $i=0,\dots,n-1$.
    In general, we use $h_b$ to denote the height 
    of chain $b$, which is also referred to as the height of block $b$.
    Evidently, the genesis block is the only block on height $0$.
\end{Def}


\begin{figure}
    \centering
    \begin{tikzpicture}
    \draw[->] (0,2.0)--(7.75,2.0)node[above] {
    time};
    \draw[-] (0,2.05)--(0,1.95) node[above] {$0$}; 
    \draw[-] (0.95,2.05)--(0.95,1.95) node[above] {$t_1$};
    \draw[-] (2.64,2.05)--(2.64,1.95)node[below] {$s$};
    \draw[-] (5.84,2.05)--(5.84,1.95)node[below] {$s+t$};
    \draw[-] (2,2.05)--(2,1.95)node[above] {$t_2$};
    \draw[-] (3.3,2.05)--(3.3,1.95)node[above] {$t_3$};
    \draw[-] (5.1,2.0)--(5.1,2.0)node[above] {$...$};
    \draw[-] (6.6,2.05)--(6.6,1.95)node[above] {$t_8$};
  \node(0) [rec]{0};
    \node[ right=.5cm of 0] (1) [blkhonest] {1};
    \draw (0)--(1);
    \node[above right=0.5cm and 0.5cm of 1] (2) [blkad] {2};
  \draw (1)--(2);
  \node[ right=1.8cm of 1] (3) [blkhonest] {3};
  \draw (1)--(3);
      \node[right=2cm of 2] (4) [blkhonest] {4};
\draw (2)--(4);
  \node[ right=0.9cm of 3] (5) [blkhonest] {5};
  \draw (3)--(5);
   \node[below right=0.75cm and 0cm of 5] (6) [blkhonest] {6};
   \draw (1)--(6);
 \node[right=1.0cm of 5] (7) [hidad] {7};
 \draw (5)--(7);
  \node[right=1.5cm of 4] (8) [blkhonest] {8};
 \draw (4)--(8);
   \node[ right=0.6cm of 7] (9) [hidhonest] {9};
 \draw (7)--(9);
\node[right=0.8cm  of 8] (10) [hidhonest] {10};
 \draw (8)--(10);
 \end{tikzpicture}
     \caption{An example of 
     blocks and chains.
     Nonpublic blocks are shaded. Each adversarial block is marked by a pair of horns on top.}
     \label{fig:fig1}
\end{figure}
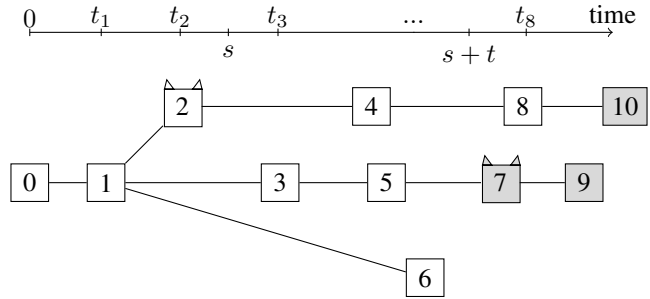

Fig.~\ref{fig:fig1} illustrates an example of chains formed by blocks 0 to 10 along with a time axis. 
It depicts a snapshot 
right after $t_{10}$ when block $10$ is mined.
Each block is horizontally positioned at its mining time and
is connected by a line to its unique parent to its left and to its children to its right (if any).
Chain $2$'s height is $h_2=2$; chain $5$'s height is $h_5=3$. We say chain $2$ is lower than chain $5$.
Even though block $6$ is mined after block $5$, chain~5 is higher than chain $6$.
Blocks $9$ and $10$ or chains $9$ and $10$ are both the highest at $t_{10}$. 

\begin{Def}[public and public height]
\label{def:public}
    A 
    chain 
    may become public at a specific 
    time 
    strictly after its creation 
    and will remain public at all times afterward.
    If a chain is public, all of its prefix chains must also be public.
   As a convention, the genesis block is public at time 0 and block 1 is mined afterward, i.e., $t_1>0$.
   At time $t$, the height of a highest public chain is referred to as the public height. 
\end{Def}

We use the notion of being public to capture the situation that a chain is known to all honest nodes.
Every chain 
consists of a public prefix chain 
followed by some nonpublic blocks (possibly none).
In Fig.~\ref{fig:fig1}, the snapshot at $t_{10}$, we shade blocks which are nonpublic.
Chain~5 is 
the public prefix of chain 9.  If chain 9 becomes public at a later time, then 
chain 7 must also be public by that time.

\begin{Def}[credible]
\label{def:credible}
    A chain mined by time $t$ is said to be credible at 
    $t$, or $t$-credible,
    if it is no lower than the public height at $t$.
\end{Def}

Being public or not is an inherent 
attribute of a chain, 
while credibility is a different attribute 
determined by the set of public chains.
Either public or not, a chain is $t$-credible if it is no shorter than any public chain at time $t$.
In the snapshot at $t_{10}$ depicted in Fig.~\ref{fig:fig1}, chain 8 is the highest public chain. Chains 4, 5, 6 are lower than chain 8, so they are not $t_{10}$-credible. Chains 7, 8, 9 and 10 are all $t_{10}$-credible.
If chain 10 becomes public at a later time, then chain 7 will no longer be credible, regardless of whether it is public at that time.

\subsection{Confirmation and Safety}
\label{s:confirmation}



\begin{Def}[confirmation by time] 
\label{def:conf_time}
    Let the confirmation time be fixed and denoted by $t>0$.
    Given a time $s$, if a chain $g$ is mined by $s$ and 
    is a prefix of a $(s+r)$-credible chain $c$ with 
    $r\geq t$, 
    then we say chain $c$ confirms chain $g$ by time of $t$.
\end{Def}

An example is depicted
in Fig.~\ref{fig:fig1} with $s$ and $s+t$ marked on the time axis. Suppose
chain 1 is mined by time $s$. 
Since chain~7 is credible at $t_7>s+t$, it confirms chain~1 by time of $t$. Likewise, chain~8 confirms chain~2 by time of $t$.

\begin{Def}[confirmation by depth] 
\label{def:conf_depth}
    Let the confirmation depth be a fixed natural number 
    denoted by $k$. 
    If a block $g$ is included in 
    a credible chain $b$ with $h_b\ge h_g+k-1$,
    then we say chain $b$ confirms block $g$ by depth of $k$ or block $g$ is $k$-confirmed by chain $b$.
\end{Def}

In Fig.~\ref{fig:fig1}, for example, block 2 is 3-confirmed by chain~8, and block~3 is 3-confirmed by chain 7 and 4-confirmed by chain 9.
With confirmation by either time or depth, we can define the violation of a 
block's safety as follows:

\begin{Def}[safety violation of a 
block]
\label{def:safety_vio}
    If there exist $r$ and $r'$ with $r' \geq r$  
    as well as
    \begin{itemize}
        \item an r-credible chain $c$ 
        which confirms block $g$,
        and 
        \item an 
        $r'$-credible chain which does not include 
        block $g$, 
    \end{itemize}
    then we say 
    block~$g$'s safety is violated.
\end{Def}

In Fig.~\ref{fig:fig1}, 
block~3 is 3-confirmed by $t_7$-credible chain~7, but is not included in chain~8, which is $t_8$-credible.
Thus block~3's safety is violated.

\subsection{Adversary and Propagation Delay}
\label{s:adversary}

\begin{Def}[types of blocks]
\label{def:types}
    A block other than the genesis block is either an A-block or an H-block.  The parent of an H-block must be credible when the H-block is mined.
\end{Def}

An A-block (resp.\ H-block) is a mathematical concept that represents a block mined by an adversarial (resp.\ honest) miner. According to Definitions~\ref{def:credible} and~\ref{def:types}, an H-block must be higher than all public 
chains at the time of its mining. There is no such 
constraint on an A-block. In fact, this is the only constraint on block heights.
In Fig.~1, we distinguish each A-block using a pair of horns on top.
Blocks 2 and 7 are A-blocks, while all the other 
blocks are H-blocks.


The propagation delays of 
blocks and chains depend on the network conditions and may also be manipulated by the adversary.  To establish a strong safety guarantee, we adopt the widely used $\Delta$-synchrony model, which is defined by the following constraint:

\begin{Def}[chain delay bound $\Delta$]
\label{def:delay}
    If block $b$ (mined at $t_b$) is an H-block, then chain~$b$ must be public at all times strictly\footnote{The H-block $b$ is nonpublic at $t_b$ according to Definition~\ref{def:public}.  In the degenerate case of $\Delta=0$, it becomes public right after $t_b$.} 
    after $t_b+\Delta$.
\end{Def}


\begin{Def}[Poisson mining processes]
    A-blocks 
    are mined 
    according to a Poisson point process with rate $a$ (blocks per second).  H-blocks 
    are mined 
    according to a Poisson point process with rate $h$ (blocks per second).  The two Poisson mining processes are independent.
\end{Def}

Let $A_{s,t}$ denote the number of A-blocks mined during $(s,t]$.  Let $H_{s,t}$ denote the number of H-blocks mined during $(s,t]$.  Evidently, $(A_{0,t}+H_{0,t})_{t\ge0}$ is a Poisson point process with rate 
\begin{align}
    \lambda = a+h.    
\end{align} 
As a convention, let $A_{s,t}=H_{s,t}=0$ for all $s>t$.  
Since the probability that any two or more blocks are mined at exactly the same time is equal to zero, we assume without loss of generality that all blocks are mined at different times. Hence $0=t_0<t_1<t_2<\dots$.

\section{Model Sufficiency}
\label{s:sufficiency}

\subsection{Sufficiency}

The prevailing model in this line of work (see, e.g.,~\cite{garay2015bitcoin, pass2017analysis, pass2017rethinking, garay2017bitcoin, dembo2020everything, gazi2020tight, saad2021revisiting, neu2022longest, gazi2022practical, kiffer2023security}) is based on the 
{\it view-based model}. Such a model defines a protocol as a set of algorithms executed in an environment that activates 
$n$ parties, which may be either honest or corrupted, each with its own view~\cite[Sec.~2]{pass2017analysis}. The protocol's execution is no simpler than a system of $n$ interactive Turing machines. 
The view-based model 
often assumes discrete time steps in execution.

The model proposed in Sec.~\ref{s:model} of this paper is entirely self-contained and much simpler than the view-based model. In contrast to previous works as well as~\cite{li2021close, guo2022bitcoin}, our model intentionally excludes miners and their views in definitions. In particular, with well-defined block mining processes, we are not concerned with the number of miners; neither do we keep track of which blocks an individual miner has mined or seen. Instead, our model 
keeps track of {\it the set of public chains} at every point in time, which 
represents {\it the weakest possible view} to any miner. 
In particular, every honest view is a superset of this weakest view, and conversely, the adversary in the view-based model can ensure at least one honest miner 
holds that weakest view.

Relative to the view-based model under $\Delta$-synchrony, our model has no loss of generality in the sense that, for a given realization (or sample path) of the A- and H-block mining processes, if the view-based model can produce a certain set of chains, so can our model. This is because both models respect the same constraints for chain creation described in Definitions~\ref{def:types} and~\ref{def:delay}.
For the latency-security analysis, our model is as powerful as the view-based model. 
To elaborate, consider the event in the view-based model where, first, a chain in an honest view confirms a block on height $h$, and then another chain in an honest view confirms a different block on the same height, creating a safety violation. Our model can produce the same pair of chains, each of which is credible to the weakest view at the time it confirms a conflicting block.
Conversely, if our model produces a pair of conflicting credible chains, then those chains can be used to convince two honest parties in the view-based model to confirm different blocks on the same height. Once we realize the role of the weakest view, there is no need to define individual 
views and their complicated 
interactions.

\subsection{Operational Meaning}

For concreteness and practicality, we relate the definitions in Sec.~\ref{s:model} to the operations of an individual miner (or node), who
assembles transactions into blocks and then chains by expending proof of work.
If the node is honest, it
follows the longest-chain fork-choice rule and immediately share new blocks via the peer-to-peer network; 
otherwise, it can use an arbitrary (possibly randomized) attack 
strategy. 
At time $t$, 
every node has its own view, which includes 
all transactions, blocks, and chains the node has seen by $t$.
The operations of this 
node is based on the entire history of this node's view, which
can only increase with time.

In general, an honest node will provide goods and services for a transaction only after it {\em commits} the transaction.
As committing a transaction has no bearing on chain production, nodes are free to choose their individual commit rules. A popular choice is to commit after $k$ block confirmations, where each node can select its own depth number $k$. In operations, after an honest node sees a new block $g$ extending one of the highest chains in the node's view, the node commits all new transactions in all blocks in chain $g$ as soon as any chain $b$ with height $h_b\ge h_g+k-1$ becomes a highest chain in its view and has chain $g$ as a prefix.\footnote{In case both block $g$ and its child block $c$ enter the node's view at once, the node can commit transactions in chain $g$ first according to this rule, and then commit additional transactions in chain $c$ in the same manner.} We emphasize that block $g$ is $k$-confirmed by chain $b$ by Definition~\ref{def:conf_depth}. In other words, committed transactions must be $k$-confirmed.


An alternative commit rule is by time of $t$: A node focuses on a block $g$ in its view at (the current) time $s$. If there ever is a time $r\ge t$, such that chain $g$ is a prefix of a highest chain in this node's view at time $s+r$, 
then the node commits all new transactions contained in chain $g$. By Definition~\ref{def:conf_time}, since this highest chain at time $s+r$ is credible,
chain $g$ is confirmed by time of $t$.

With confirmation by either depth or time, once we have a safety guarantee established rigorously for sufficiently confirmed chains in the mathematical model in Sec.~\ref{s:model}, the honest node enjoys the same operational safety guarantee for every committed transaction. There is in fact the subtlety that block $g$ is anchored to a current view at a point in time chosen by the honest node, so that the adversary cannot directly carry out an adaptive attack of the following form: Mine in private to gain a sufficiently high lead, then pick a block to attack with certainty to cause a violation.

\subsection{Delay Variations}

The seemingly rigid $\Delta$-synchrony model is quite practical. First, ``light'' nodes do not mine blocks or fork chains, so the delays they perceive are inconsequential to ledger safety as long as honest miners satisfy the delay requirement.  Second, if some blocks of honest miners suffer longer delays than $\Delta$, we simply regard those blocks as adversarial blocks to preserve $\Delta$-synchrony.  In addition, potential attacks like Erebus~\cite{Erebus} and SyncAttack~\cite{SyncAttack} seek to partition the peer-to-peer network. 
As long as those delay violations are independent of the block mining times, the honest and adversarial mining processes remain to be Poisson.  If block propagation delays have a long-tailed distribution~\cite{decker2013information, fechner2022calibrating}, we can use a delay bound that applies to a high percentage of blocks.
The main results here continue to apply with all out-of-sync mining power treated as being adversarial.

\section{Main Results}\label{s:theorem}


Recall that $a$ and $h$ represent as the adversarial and honest mining rates, respectively, in the number of 
blocks mined per second.
For confirmation by either time or depth, the following negative result is well understood (see, e.g.,~\cite{pass2017analysis, dembo2020everything, gazi2020tight}).

\begin{theorem}[unachievable fault tolerance]
\label{th:unsafe}
    If the mining rates and the delay upper bound satisfy
    \begin{align} \label{eq:a>}
        \frac1a < \Delta + \frac1{h},
    \end{align}
    then every 
    block is unsafe in the sense that the private-mining attack 
    will violate its safety with probability 1.
\end{theorem}
For completeness, we prove Theorem~\ref{th:unsafe} in Appendix~\ref{a:unsafe}.

\subsection{Bounds for Given Confirmation Depth}
\label{subsec:bound_depth}


Let us define
  \begin{align}\label{eq:b()}
        b(\alpha,\eta)
        = \left( \frac{e^{\frac{5}{2}\gamma}\gamma(\eta-\alpha-\alpha \eta)}
        {\alpha \eta e^{\gamma} - (\alpha+\gamma)(\eta-\gamma)} \right)^2
        + 1
        - \frac{\eta-\alpha-\alpha\eta}{\eta e^\alpha},
    \end{align}
and
\begin{align}\label{eq:c()}
\begin{split}
    c(\alpha,\eta)
    =
    \left( 
    \ln \left( 1 + \frac{\gamma}{\alpha} \right)
    + \ln \left(1-\frac{\gamma}{\eta} \right)
    - \gamma\right) \cdot \left( 1 - \ln \left(1+\frac{\gamma}{\alpha} \right)
    \bigg/
    \ln \frac{\eta e^{\alpha}(e^{\eta}-1)}
    {\eta+\eta e^{\alpha}(e^{\eta}-1)-\alpha(1+\eta)}
    \right)^{-1}
\end{split}
\end{align}
    where $\gamma$ in~\eqref{eq:b()} and~\eqref{eq:c()} is a shorthand for
    \begin{align}\label{eq:mu}
        \gamma(\alpha,\eta) =\frac{2-\alpha+\eta-\sqrt{4+  (\alpha+\eta)^2}}{2} .
    \end{align}

The following latency-security trade-off is the first set of non-trivial safety guarantees that apply to all parameter values up to the ultimate fault tolerance.

\begin{theorem}[Safety guarantee 
for given confirmation depth]
\label{th:Cher_upper}
   Suppose the mining rates and the delay upper bound satisfy 
   \begin{align} \label{eq:repeat a<}
       \frac1a > \Delta + \frac1h .
   \end{align}
    Given a sufficiently large $s$, 
    the safety of a 
    block extending an $s$-credible chain and being $k$-confirmed 
    can be violated with a probability no higher than
    \begin{align}\label{eq:CherUpper}
        b(a\Delta,h\Delta) e^{-c(a\Delta,h\Delta) \cdot k}
    \end{align}
where $b(\cdot,\cdot)$ and $c(\cdot,\cdot)$ are defined in~\eqref{eq:b()} and~\eqref{eq:c()}, respectively.
\end{theorem}

Theorem~\ref{th:Cher_upper} is proved in Appendix~\ref{s:Cher_upper}, in which the following result is established as a byproduct in Appendix~\ref{s:cor:tolerance}.

\begin{corollary}[Achievable fault tolerance]
\label{cor:tolerance}
    As long as~\eqref{eq:repeat a<} holds, 
    every 
    block 
    is arbitrarily safe after sufficiently long confirmation latency (either in depth or time).
\end{corollary}

The pair of opposing inequalities~\eqref{eq:a>} and~\eqref{eq:repeat a<} determine the exact boundary between the asymptotically achievable and unachievable latency-security trade-offs.
The right hand side of~\eqref{eq:a>} and~\eqref{eq:repeat a<} is exactly the expected inter-arrival time of the following type of blocks:

\begin{Def}[Pacer]
We say the genesis block is the $0$-th pacer. After the genesis block, every pacer is the first H-block mined at least $\Delta$ time after the previous pacer. We use $P_{s,t}$ to denote the number of pacers mined in $(s,t]$.
\end{Def}
By definition, every pacer becomes public before the next pacer 
arrives, and hence all pacers are on different heights. 
Since the expected pacer inter-arrival time is $\Delta + 1/\myh$, pacers arrive at the rate of $\myh/(1+\myh\Delta)$ in the long run.
As a consequence, the number of heights H-blocks occupy grows at least at that rate.
This approach enables us to use the growth rate of pacers as a lower bound for the height increase of H-blocks.

Let us introduce the following probability density function (pdf):
\begin{align} \label{eq:interpacertime}
    f(t) = \frac1\myh e^{-\myh(t-\Delta)} 1_{\{t>\Delta\}} .
\end{align}

\begin{lemma}[{\cite[Theorem 3.4]{Kroese_difference}}]
\label{lm:LS_trans}
    Let $A_r$ denote a Poisson point process with rate $\mya$.  Let $Q_r$ denote a renewal process whose inter-arrival time has a pdf given by~\eqref{eq:interpacertime}.
    Then the probability generating function of the maximum difference between the two processes
    \begin{align} \label{eq:defY}
        Y = \sup_{t\ge0} \{ A_t - Q_t \}
    \end{align}
    is equal to
    \begin{align}\label{eq:E(r)}
        \mathcal{E}(r)
        = \frac{(1-r)(\myh - \mya -\myh\mya \Delta)}{\myh - e^{(1-r)\mya \Delta}(\myh+\mya-\mya r)r} .
    \end{align}
\end{lemma}

We denote the probability mass function (pmf) of $Y$ as
\begin{align} \label{eq:e(i)}
    e(i) &= \frac{1}{i!} \mathcal{E}^{(i)}(0) , \qquad i=0,1,\dots
\end{align}
where 
$\mathcal{E}^{(i)}(0)$ represents the 
the $i$-th derivative of 
$\mathcal{E}(\cdot)$ 
evaluated at $r = 0$.

Throughout this paper, we let
\begin{align} 
    \epsilon &= 1-e^{-h\Delta} \label{eq:epsilon}
\end{align}
and
\begin{align}
    \delta &= 
    \left( 1-a\frac{1+h\Delta}{h} \right)^+ e^{-a\Delta} . \label{eq:delta}
\end{align}
Let us define
\begin{align}
    F^{(k)}(n) =     1- (1-\delta) \left(\frac{\epsilon}{\epsilon+\delta-\epsilon\delta}\right)^{n+1} \bar{F}_b(n;k, (\epsilon+\delta-\epsilon\delta)) \label{eq:Fkn}
\end{align}
where $\bar{F}_b(\cdot;k,p)$ denotes the complementary cumulative distribution function (cdf) of the binomial distribution with parameters $(k,p)$.
We also define
\begin{align} \label{eq:Fn}
    F(n)
    &= \lim_{k\to\infty} F^{(k)}(n) \\
    &=  1-(1-\delta)
    \left( \frac{\epsilon}{\epsilon+\delta-\epsilon\delta} \right)^{n+1} .
\end{align}

\begin{theorem}[Finer upper bound for the violation probability for given confirmation depth]\label{th:upper_depth}
   Given a sufficiently large $s$,
    the safety of a 
    block that extends an $s$-credible chain and is
    confirmed by the depth of $k$ can be violated with a probability no higher than
    \begin{align} \label{eq:upper_depth}
       \inf_{n}\bigg\{2-&\sum_{i=0}^{k-n}e(i)\sum_{j=0}^{k-n-i}e(j) \int_{0}^{\infty}F_1(k-n-i-j;a(t+3\Delta)) \cdot f_{2}(t-(k)\Delta;k,h)  \, dt -F^{(k)}(n-1) \bigg\} 
    \end{align}
    where $F^{(k)}$ is defined in~\eqref{eq:Fkn}, $e(i)$ is given by~\eqref{eq:e(i)}, $F_1(\cdot;\lambda)$ denotes the cdf of the Poisson distribution with mean $\lambda$, and $f_2(\cdot;k,h)$ denotes the pmf of the Erlang distribution with shape parameter $k$ and rate parameter $h$.
\end{theorem}

Theorem~\ref{th:upper_depth} is proved in Appendix~\ref{a:upper_depth}.
                                 
To ease computation, the infinite integral in~\eqref{eq:upper_depth} could be replaced by only integrating to a finite time limit, where the larger the limit, the tighter the upper bound.

A well-understood attack is the private-mining attack, which we will use to lower bound the probability of safety violation.  The attack is defined precisely as follows.

\begin{Def}[Private-mining attack on a target 
block]
\label{def:private-mining}
    The adversary delays the propagation of every H-block maximally (for $\Delta$ seconds). Starting from time 0, the adversary always mines in private to extend the highest chain which does not include the target 
    block.
\end{Def}

Recall that Definition~\ref{def:safety_vio} allows nonpublic credible chains to cause a violation. However, operationally, the adversary can carry out a private-mining attack and make the private chain public as soon as both of the following conditions hold:
\begin{inparaenum}
    \item the 
    block has been confirmed by a chain of H-blocks; and
    \item the private chain which does not contain the 
    block is credible.
\end{inparaenum}
In other words, as long as a safety violation occurs by Definition~\ref{def:safety_vio}, the adversary creates a violation operationally.

\begin{theorem}[Lower bounding the violation probability for given confirmation depth]
\label{th:lower_depth}
    Given a sufficiently large $s$, 
    the private-mining attack violates the safety of a target 
    block 
    which 
    extends an $s$-credible chain with at least the following probability:
\begin{align}\label{eq:lower_depth}
    1- \sum_{i=0}^{k} \left(1-\frac{a}{h}\right) \left(\frac{a}{h}\right)^i \cdot \sum_{j=0}^{k-i}e(j) \int_{0}^{\infty} F_1(k-(1+i+j);at) \cdot f_{2}(t-k\Delta;k,h)\, dt
\end{align}
where $e(\cdot)$, $F_1$, $f_2$ are defined as in Theorem~\ref{th:upper_depth} and $f_1(\cdot; \lambda)$ denotes the pmf of the Poisson distribution with mean $\lambda$.
\end{theorem}

Theorem~\ref{th:lower_depth} is proved in Appendix~\ref{a:lower_depth}.

\subsection{Throughput and Related Trade-offs}

An ideal protocol design should select parameters to yield high throughput, fast confirmation, and security against adversarial miners.
We refer to the average number of blocks mined during the maximum propagation delay, $\lambda\Delta$, as the natural fork number.
Throughout the paper, let $\beta$
denote the maximum tolerable fraction of adversarial mining. To operate within fault tolerance~\eqref{eq:a<},
the natural fork number must be 
bounded by
\begin{align} \label{eq:fork<}
    \lambda\Delta < \frac1{\beta} - \frac1{1-\beta} .
\end{align}
We remark that it is more convenient to discuss throughput in terms of $\beta$ and the total mining rate $\lambda$ in lieu of the adversarial and honest mining rates ($a$ and $h$).

To achieve a high throughput (in KB/second or transactions/second), the designer may select a high mining rate and/or a large block size, where the latter also implies a large block propagation delay.
The optimistic throughput, where all miners work to grow the highest chains and natural forks are orphaned, is equal to $\lambda B/ (1+\lambda \Delta)$, where $B$ represents the maximum allowable size of a block, ensuring that the propagation delay does not exceed the upper limit $\Delta$.
The following result is a straightforward consequence of~\eqref{eq:fork<}.

\begin{theorem}
    If the blocks are no larger than $B$ KB, the block propagation delays can be as large as $\Delta$ seconds, in order to tolerate adversarial mining fraction up to $\beta$, the throughput cannot exceed
    \begin{align}    \label{eq:throughput<}
        \frac{1-2\beta}{1-\beta-\beta^2}
    \frac{B}{\Delta} \quad\text{KB/s}.
    \end{align}
\end{theorem}


Empirically, the block propagation delay is approximately an affine function of the block size~\cite{decker2013information, fechner2022calibrating} (even with a coded design~\cite{zhang2022speeding}).
That is, the propagation delay bound can be modeled as
\begin{align}\label{eq:throughput_delta}
    \Delta(B) = 
    \frac{B}r + \nu
\end{align}
where $B$ is the block size in KB and 
$r$ and $\nu$ are constants.
Here, $r$ represents a fundamental limitation on the end-to-end transmission rate of the peer-to-peer network.  Based on~\eqref{eq:throughput<}, this rate and the desired fault tolerance determine an upper bound for the optimistic throughput:
\begin{align} \label{eq:throughput_r}
    \frac{\lambda B}{1+\lambda\Delta}
    <
    \frac{1-2\beta}{1-\beta-\beta^2} \, r.
\end{align}


With a close upper bound on the probability of safety violation as a function of the system parameters and the confirmation latency, e.g., in the form of
    $p( \lambda \Delta, \beta, k )$,
a node 
may select the confirmation depth $k$ to achieve a desired safety. 
In particular, the creator of a highly-valued transaction 
will most likely tolerate a longer-than-usual confirmation latency. The theorems in this paper allow this choice to be made much more accurately 
than before.

It is also meaningful to consider optimizing the throughput by selecting parameters that guarantee a certain safety at a desired confirmation latency (in time) under a worst-case assumption of adversarial mining.
In particular, we can formulate the following problem:
\begin{subequations}
\label{eq:maximize}
\begin{align}
    \text{maximize}_{\lambda,B,k}
    \quad & \frac{ \lambda B }{ 1 + \lambda\, \Delta(B) }  \label{eq:throughput} \\
    \text{subject to}
    \quad 
    & p( \lambda\, \Delta(B), \beta, k ) \le q \label{eq:p<=q} \\
    & k \leq \frac{(1-\beta)\lambda}{1+(1-\beta)\lambda\Delta}d \label{eq:k<=}
\end{align}
\end{subequations}
where $q$ denotes the maximum tolerable probability of violation and $d$ denotes the expected confirmation time. Here~\eqref{eq:k<=} constrains $k$ such that at the slowest growth rate, gaining $k$ heights takes no more than $d$ units of time on average. In Sec.~\ref{s:numerical}, we give such a numerical example under assumption~\eqref{eq:throughput_delta} with empirically determined constants.



\subsection{Bounds for Given Confirmation Time}
\begin{theorem}[safety guarantee 
for given confirmation time]\label{th:upper_time}
    Suppose the mining rates and the delay upper bound satisfy~\eqref{eq:a<}.
    Given a sufficiently large $s$,
    the safety of a 
    block mined by time of $s$ and is confirmed by time \sj{of} $t$ can be violated with a probability no higher than
    \begin{align}\label{eq:upper}
        \inf_{n\in \{1,2,\dots\}}\{1-F(n-1)+G(n,t)\},
    \end{align}
    where $F$ is defined in~\eqref{eq:Fn}, and 
    \begin{align} \label{eq:Gnt}
    G(n,t) = \sum_{i=0}^\infty \sum_{j=0}^\infty \sum_{k=0}^\infty 
    f_1(i;\mya (t+2\Delta)) e(j) e(k)  \times (1-F_2(t-(i+j+k+n+2)\Delta;i+j+k+n,\myh))
    \end{align}
    where $e(\cdot)$ and $f_1$ are defined the same as in Theorem~\ref{th:lower_depth}, and $F_2(\cdot ; n, \myh)$ denotes the cdf of the Erlang distribution with shape parameter $n$ and rate parameter $\myh$.
\end{theorem}

Theorem~\ref{th:upper_time} is proved in Appendix~\ref{a:upper_time}.
In lieu of the infinite sums in~\eqref{eq:Gnt}, a computationally feasible bound is obtained by replacing $G(n,t)$ with its following upper bound
\begin{align} \label{eq:Gnt_}
\begin{split}
    \overline{G}(n,t) &= 1-
    F_1(i_0;\mya (t+2\Delta)) F_E(j_0) F_E(k_0) \\
      & + 
    \sum_{i=0}^{i_0}\sum_{j=0}^{j_0}\sum_{k=0}^{k_0}
    f_1(i;\mya (t+2\Delta))e(j)e(k)
    \times  (1-F_2(t-(i+j+k+n+2)\Delta;i+j+k+n,\myh))
\end{split}
\end{align}
where $F_1$, $F_E$ are the corresponding cdfs of variables with pmf of $f_{1}(i; \mya (t+2\Delta))$, $e(j)$ in~\eqref{eq:e(i)}, respectively.
And $i_0$, $j_0$ and $k_0$ are positive integers. The larger these integers are, 
the tighter the upper bound. 

For completeness, we include the following converse result, which is a variation of~\cite[Theorem 3.6]{li2021close} due to the slightly different definition of safety violation here.

\begin{theorem}[Lower bound 
for given confirmation time]
\label{th:lower_time}
Given a sufficiently large $s$,
the private-mining attack violates the safety of a target 
block mined by time $s$ 
and confirmed by time of $t$
with at least the following probability: 

\begin{align}\label{eq:lower}
\left(1-\frac{\mya}{\myh}\right)  e^{(\myh-\mya) t} \sum_{j=0}^{\infty} \sum_{k=0}^{\infty} e(j)\left(\frac{\mya}{\myh}\right)^{k} F_{1}(k ; \myh t)  \times(1-F_{2}(t-(j+k) \Delta ; j+k+1, \myh))- e^{-h(t-\Delta)}.
\end{align}

\end{theorem}

Theorem~\ref{th:lower_time} is proved in Appendix~\ref{a:lower_time}.
Evidently, if the infinite sums in~\eqref{eq:lower} are replaced by partial sums, 
the resulting (tighter) security level also remains unachievable.

\section{Essential Elements of the Analysis}
\label{s:essential}

In this section, we introduce the most essential technical elements, discuss intuitions, and provide a sketch of the statistical analysis.

\subsection{The Balancing Process}
\begin{Def}[Agree and disagree]
    Two chains that are no lower than a target block 
    are said to agree on the target block if either they both include 
    the target block or neither includes it; 
    otherwise, they are said to disagree on the target block.
\end{Def}

Agreement on a target 
block is an equivalence relation, which 
divides all blocks 
higher than the target block into two equivalence classes. In Fig.~\ref{fig:fig1},
agreement about target block 3 divides blocks on 
height 2 and above into two equivalence classes, namely, blocks $\{3,5,7,9\}$ and blocks $\{2,4,6,8,10\}$.

\begin{Def}[balancing candidate, candidacy]
Given a target 
block, if an H-block (denoted as block $b$) is the first block mined on height $n$ and disagrees on the target 
block with any height-$(n-1)$ block mined before block $b$ becomes public, then block $b$ is said to be a balancing candidate, or simply a candidate. The period from the first time the 
disagreement arises to the time block $b$ becomes public is referred to as its candidacy. Evidently, if any disagreeing height-$(n-1)$ block is mined before block $b$, its candidacy begins from $t_b$.
\end{Def}

\begin{Def}[balanced 
height]
    If 
    a disagreeing H-block is mined on the candidate’s height and no 
    A-block is mined on that height, then the candidate is said to be balanced. Accordingly, its height is said to be a balanced height. 
\end{Def}

Given a target block $g$, we use 
$X_{s,t}^g$ to denote the total number of balanced candidates 
created during $(s,t]$. We also define $X^g_t=X^g_{t_g,t}$ and
\begin{align} \label{eq:defX}
    X^g= \lim_{t\to\infty} X_{t}^g .
\end{align}

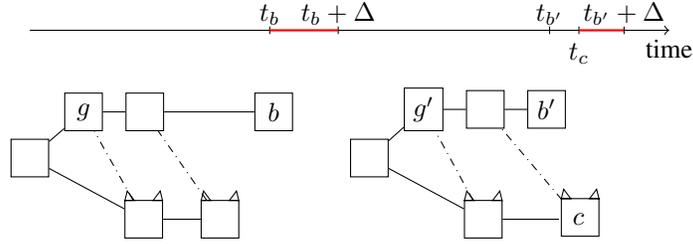
\begin{figure}
\begin{center}
    \begin{tikzpicture}
    \draw[->] (0,1.7)--(8.5,1.7)node[below] {
    time};
      \node(0) [rec]{};
        \node[below right=0.3cm and 1cm of 0] (1) [blkadempty] {};
          \draw (0)--(1);
    \node[right=0.5cm of 1] (2)[blkadempty] { };
      \draw (1)--(2);
    \node[above right= 0.1  and 0.2cm of 0](3)[blkhonest]{$g$};
    \node[right=0.3 of 3](4)[blkhonest]{};
       \draw (0)--(3);
        \draw (3)--(4);
        \node [right = 1.2 of 4](5)[blkhonest]{$b$};
        \draw (4)--(5);
        \draw[-] (3.19,1.75)--(3.19,1.65) node[above] {$t_b$};
         \draw[-] (4.1,1.75)--(4.1,1.65) node[above] {$t_b+\Delta$};
    \draw[-,thick,red] (3.19,1.7)--(4.1,1.7);

    \node [right = 4cm of 0](0')[rec]{};
    \node[below right=0.3cm and 1cm of 0'] (1') [blkadempty] { };
          \draw (0')--(1');
    \node[right=0.75cm of 1'] (2')[blkad] {$c$};
      \draw (1')--(2');
    \node[above right= 0.1  and 0.2cm of 0'](3')[blkhonest]{$g'$};
    \node[right=0.3 of 3'](4')[blkhonest]{};
       \draw (0')--(3');
        \draw (3')--(4');
        \node [right = 0.3 of 4'](5')[blkhonest]{$b'$};
        \draw (4')--(5');
        \draw[-] (7.3,1.75)--(7.3,1.65) node[below] {$t_c$};
        \draw[-] (7.9,1.75)--(7.9,1.65) node[above] {$t_{b'}+ \Delta$};
        \draw[-](6.91,1.75)--(6.91,1.65) node[above] {$t_{b'}$};
        \draw[-,thick,red] (7.3,1.7)--(7.9,1.7);
        \draw[dash dot] (1)--(3);
        \draw[dash dot] (2)--(4);
        \draw[dash dot] (1')--(3');
        \draw[dash dot] (2')--(4');
     \end{tikzpicture}
    \end{center}
    \caption{Balancing candidates $b$ and $b'$ for target block $g$ and $g'$, respectively.  
    Blocks on the same height are connected by dashed lines.  Left: block $b$ becomes a balancing candidate upon its mining; its candidacy is $[t_b,t_b+\Delta]$. Right: block $b'$ becomes a candidate after a lower block $c$ is mined, in disagreement with block $b'$. 
    Block $b'$ has a shorter candidacy $[t_c,t_{b'}+\Delta]$.
    }
    \label{fig:candidacy}
\end{figure}




A candidacy is a small window of opportunity to balance the corresponding candidate (and thus create a balanced height). 
Figs.~\ref{fig:candidacy} illustrates two different balancing candidacies at different times. 
Examples of balancing can be found in 
Fig.~\ref{fig:fig1}.
If block 3 is the target, blocks 4 and~5 balance each other on height 3; blocks 9 and 10 balance each other on height 5. No other heights are balanced. In particular, while H-blocks 3 and 6 are on the same height, they do not balance each other due to the existence of block 2 as an A-block on the same height.
We shall show that the total number of heights that can be balanced 
is essentially stochastically dominated by a geometric random variable.

\subsection{Violation Event under Confirmation by Depth}

\begin{Def}[Pre-mining lead] \label{def:premining}
    The pre-mining lead (of the adversary) at time $t$, denoted by $L_t$, is the height of the highest block mined by time $t$ minus the 
    public height at time $t$.
\end{Def}

The notion of pre-mining lead is to account for the potential lead the adversary may already have when the target 
block is included by a blockchain. In this work, $s$ is given so that the adversary cannot mount an adaptive attack (by influencing $s$) depending on the lead. Nonetheless, the adversary can try to create and maintain some lead (if possible), which is essentially characterized by a birth-death process.

\begin{Def}[Jumper]
\label{def:jumper}
    The first H-block mined on its height is called a jumper. We use $J_{s,t}$ to denote the number of jumpers mined in $(s,t]$.
\end{Def}



In the remainder of this section, we assume $s$ is a given time.
We focus on the case of confirmation by depth $k$.
It is easy to see that the height growth of H-blocks is no less than the number of jumpers mined. Intuitively, for a block's safety to be violated, there must exist an interval, during which the number of jumpers mined is fewer than the sum of the pre-mining lead, the number of balanced heights, and the number of A-blocks mined. This is made precise in the following lemma.


\begin{lemma}
\label{lm:errorevent}
    Given $s$, let $\tau$ denote the arrival time of the first H-block that is 
    at least $k$ higher than the public height at time $s$. If a block $g$ extends an $s$-credible chain
    and its safety is violated, then the following 
    must hold: 
    \begin{align}
        L_s + X_{t_f +\Delta}^g + \sup_{d \geq \tau} \{A_{s,d}-J_{\tau,d-\Delta}\} \geq k 
    \label{eq:violates_aftertau}
    \end{align}
    where block $f$ is the first block mined on height $h_g+k-1$, whose mining time is denoted as $t_f$.
\end{lemma}

\begin{proof}
We first show that~\eqref{eq:violates_aftertau} is implied by a slightly stronger inequality:
\begin{align}
     L_s + X_{t_f +\Delta}^g + \sup_{d \geq t_f} \{A_{s,d}-J_{\tau,d-\Delta}\} \geq k 
    \label{eq:violates_aftertf},
\end{align}
where the domain of optimization $d\ge\tau$ in~\eqref{eq:violates_aftertau} is replaced by $d\ge t_f$. 
In the case of $t_f\ge\tau$, the domain in~\eqref{eq:violates_aftertau} is larger, yielding a higher supremum value, hence~\eqref{eq:violates_aftertau} holds as long as~\eqref{eq:violates_aftertf} holds. In the remaining case of
$t_f < \tau$, we have
\begin{align}
    \sup_{d \geq t_f} \{A_{s,d}-J_{\tau,d-\Delta}\}
    &= 
    \max\left\{ \sup_{d > \tau} \{A_{s,d}-J_{\tau,d-\Delta}\},  \sup_{\tau\ge d\ge t_f} \{A_{s,d}-J_{\tau,d-\Delta}\} \right\} \\
    &= 
    \max\left\{ \sup_{d > \tau} \{A_{s,d}-J_{\tau,d-\Delta}\},  A_{s,\tau} \right\}
    \label{eq:maxsup1} \\
    &=
    \sup_{d \geq \tau} \{A_{s,d}-J_{\tau,d-\Delta}\}
    \label{eq:maxsup2}
\end{align}
where in~\eqref{eq:maxsup1} and~\eqref{eq:maxsup2} we have both used the fact that $J_{\tau,d-\Delta} = 0$ for $d\leq \tau$. Hence, in this case~\eqref{eq:violates_aftertau} and~\eqref{eq:violates_aftertf} are equivalent.


The remainder of this proof is devoted to establishing~\eqref{eq:violates_aftertf}.
Here we let all agreement relationship be with respect to block $g$.
Evidently, 
block~$g$ can not be $k$-confirmed before $t_f$.
Let chain~$b$ denote a highest public chain by time $s$. Since block $g$ extends an $s$-credible chain, we have $h_g>h_b$ and hence
\begin{align}
    h_f\ge h_b+k  \label{eq:hf>=hb+k}
\end{align}
by definition of block $f$ in the lemma.
Let block $c$ denote a highest H-block mined by $t_f$.
Let chain $z$ denote a highest chain that disagrees with chain $c$ at time $t_f$.
Definitions of some frequently used notations in this proof are listed in Table~\ref{tb:blocks}.

\begin{table}
\centering
\begin{tabular}{|r|l|}
\hline
Notation & Meaning \\ \hline
$s$ & a given time for commitment rule\\
\hline
chain~$b$ & a highest public chain by $s$ \\
\hline
block~$c$ & a highest H-block mined by $t_f$ \\
\hline
block~$f$ & the first block mined on height $h_g+k-1$ \\
\hline
block~$g$ & the target block \\ 
\hline
chain~$z$ & a highest chain that disagrees with chain~$c$ at $t_f$ \\ 
\hline
$\tau$ & the arrival time of the first H-block that is at least $k$ higher than chain~$b$ \\
\hline
\end{tabular}
\caption{Frequently used notations in the proof of Lemma~\ref{lm:errorevent}.}
\label{tb:blocks}
\end{table}



We first establish the following upper bound on 
the number of jumpers mined during $(\tau,t_f]$:
\begin{align}
    J_{\tau,t_f}
    \le
    h_f-(h_b+k) . 
    \label{eq:J>h_g-1-h_b}
\end{align}
If $\tau\ge t_f$, $J_{\tau,t_f}=0$ by convention, so that~\eqref{eq:J>h_g-1-h_b} holds trivially due to~\eqref{eq:hf>=hb+k}. If $\tau<t_f$, the jumpers between $\tau$ and $t_f$ must be between heights $h_b+k+1$ and $h_f$, hence~\eqref{eq:J>h_g-1-h_b} must hold.

We establish~\eqref{eq:violates_aftertf} in two separate cases
depending on whether $h_c \leq h_b+L_s$.

\begin{itemize}
\item [Case 1)] $h_c\leq h_b+L_s$.
In this case 
no H-blocks higher than $h_b+L_s$ are mined by $t_f$, so there is at least one A-block on each height from $h_b+L_s+1$ to $h_f$ mined between $s$ and $t_f$, i.e.,
\begin{align}
    A_{s,t_f} &\ge h_f - (h_b+L_s)\\
    &\ge k+J_{\tau,t_f}-L_s\label{eq:A>k+J-L},
\end{align}
where \eqref{eq:A>k+J-L} is due to \eqref{eq:J>h_g-1-h_b}.
Since 
$X_{t_f+\Delta}^g\ge0$ 
and $\sup_{d\ge t_f}\{A_{s,d}-J_{\tau,d-\Delta}\} \ge
A_{s,t_f}-J_{\tau,t_f-\Delta} \ge
A_{s,t_f}-J_{\tau,t_f}
$,~\eqref{eq:A>k+J-L} implies~\eqref{eq:violates_aftertf}.

\item[Case 2)] $h_c > h_b+L_s$. 
We divide the proof 
in this case
into two parts (2-a) and (2-b), where the first part focuses on blocks mined during $(s,t_f]$, and the second part focuses on blocks mined after $t_f$.
\begin{itemize}
    \item 
[Part 2-a)]
Since chains~$c$ and $z$ disagree, there exist at least two disagreeing chains of height $\min\{h_c,h_z\}$ by $t_f$. This 
implies that at least one A-block is mined during $(s,t_f]$ on every height between $h_b+L_s+1$ and $\min\{h_c,h_z\}$ except those heights that are balanced by time $t_f$.
The number of balanced heights is no larger than $X_{t_f}^g$. 
In addition, since no H-block is higher than $h_c$, 
at least one A-block is mined on each height from $h_c+1$ to $h_f$.
Therefore, the number of A-blocks is lower bounded:
\begin{align}
    A_{s,t_f} \geq ( \min\{h_c,h_z\}-(h_b+L_s)-X_{t_f}^g ) + (h_f - h_c).
    \label{eq:A>...max}
\end{align}

 \item 
 [Part 2-b)] We 
now focus on blocks mined after $t_f$ in two different 
 cases depending on the relative positions of blocks $c$ and $z$.
\begin{itemize}
    \item [Case $h_z \geq h_c$.] 
    Since $\min\{h_c,h_z\}=h_c$, \eqref{eq:A>...max} reduces to 
    \begin{align}
    A_{s,t_f} &\geq h_f-(h_b+L_s)-X_{t_f}^g \label{eq:A>h_c-h_b-L-X+h_f-h_c}.
    \end{align}
    Comparing~\eqref{eq:J>h_g-1-h_b} and~\eqref{eq:A>h_c-h_b-L-X+h_f-h_c} yields
    \begin{align}
     L_s + X_{t_f}^g + A_{s,t_f} -J_{\tau,t_f}
     &\geq k .
     \label{eq:no_d_tau}
\end{align}
Noting that $X^g_{t_f} \le X^g_{t_f+\Delta}$, $J_{\tau,t_f} \ge J_{\tau,t_f-\Delta}$, and reducing the domain $d\ge t_f$ in~\eqref{eq:violates_aftertf} to $d=t_f$ reduces the supremum, we have established~\eqref{eq:violates_aftertf} based on~\eqref{eq:no_d_tau} for Case $h_z\geq h_c$.


\item [Case $h_z<h_c$.] 
    In this case, block~$c$ must agree with block~$f$; otherwise chain $z$ cannot be the highest to disagree with chain $c$. Since block $f$ is the first 
    block that is $k-1$ higher than
    block $g$, whose 
    safety is subsequently violated, 
    the first chain that is credible and also disagrees with chains~$c$ and $f$ must be created at some time $d\geq t_f$ to create a safety violation. All jumpers mined during $(t_f,d-\Delta]$ become public by time $d$, so the highest public block at time $d$ is at least as high as $h_c+J_{t_f,d-\Delta}$. Because the highest disagreeing chain at $t_f$ is no higher than $h_z$, those
    blocks of the $d$-credible disagreeing chain between heights $h_z+1$ and $h_c+J_{t_f,d-\Delta}$ must be mined during $(t_f,d]$. Moreover, after $t_f+\Delta$, a violation occurs as soon as a balancing candidate is created (before it is even balanced). Hence
    at most $X_{t_f,t_f+\Delta}^g$ of heights $(h_z+1,\dots,h_c+J_{t_f,d-\Delta})$ 
    can be balanced, whereas at least one A-block is mined on the remaining heights. Putting it altogether, we have
    \begin{align}
        A_{t_f,d}&\geq h_c+J_{t_f,d-\Delta} - h_z-X_{t_f,t_f+\Delta}^g.
    \end{align}
    Adding \eqref{eq:A>...max} 
    with $\min\{h_c,h_z\}=h_z$ and canceling $h_z$ and $h_c$ on both sides, we have
    \begin{align}
       A_{s,d}
       &\geq
       J_{t_f,d-\Delta} + h_f - (h_b+L_s) - X_{t_f+\Delta}^g \\
       &\geq
       J_{t_f,d-\Delta} + J_{\tau,t_f} + k - L_s -X_{t_f+\Delta}^g \label{eq:A>J+J} \\
        &=
        J_{\tau,d-\Delta}+k-L_s-X_{t_f+\Delta}^g,\label{eq:A>J+k-L-X}
\end{align}
where~\eqref{eq:A>J+J} is due to \eqref{eq:J>h_g-1-h_b}.
Since $\sup_{d\geq t_f}\{A_{s,d}-J_{\tau,d-\Delta}\}$ in \eqref{eq:violates_aftertf}  is greater or equal to $A_{s,d}-J_{\tau,d-\Delta}$,~\eqref{eq:A>J+k-L-X} implies~\eqref{eq:violates_aftertf}.
\end{itemize}
\end{itemize}
\end{itemize}
\end{proof}

We caution that Lemma~\ref{lm:errorevent} is very general and at the same time quite delicate. While block $b$ in the proof is the highest public block by $s$, block $g$ can be arbitrarily higher than height $h_b+1$, and may or may not be mined after $s$. In fact, even block $f$ may or may not be mined before $s$, as there is the possibility that chain $f$ already $k$-confirms block $g$, yet all of those confirming blocks including block $g$ itself are nonpublic.
We illustrate~\eqref{eq:violates_aftertau} in Lemma~\ref{lm:errorevent} using a pair of concrete chains in Fig~\ref{fig:lemma10d}.

\begin{figure}
 \centering
 \begin{tikzpicture}
  \draw[->] (0,0.5)--(8.5,0.5)node[above] {$d$};
  \draw[-] (7,0.45)--(7,0.55)node[above] {$t_f$};
  \draw[-] (6,0.45)--(6,0.55)node[above] {$\tau$};
  \draw[-] (2.9,0.45)--(2.9,0.55)node[above] {$s$};
  \node(b) [blkhonest]{$b$};
  \node[right =1.9cm of b] (0) [blkhonest]{};
  \draw (0)--(b);
  \node[right =0.3cm of 0] (1) [blkhonest]{$g$};
  \draw (0)--(1);
  \node[right =0.2cm of 1] (2) [blkhonest]{};
  \draw (2)--(1);
  \node[right =0.5cm of 2] (3) [blkhonest]{};
  \draw (2)--(3);
  \node[right =0.5cm of 3] (4) [blkhonest]{};
  \draw (4)--(3);
  \node[right =0.5cm of 4] (5) [blkhonest]{$c$};
  \draw (4)--(5);
  \node[right =0.5cm of 5] (6) [hidhonest]{};
  \draw (6)--(5);
  \node[below=0.5cm of b](01) [blkad]{$e$};
  \draw[dashed] (b)--(01);
  \node[right=0.7cm of 01] (02) [blkadempty]{};
  \draw (01)--(02);
  \node[right=0.5cm of 02] (03) [blkadempty]{};
  \draw (03)--(02);
  \draw[dashed] (1)--(03);
  \node[right =0.7cm of 03] (04) [blkhonest]{};
  \draw (03)--(04);
  \node[right=0.5cm of 04] (05) [hidad]{$z$};
  \draw (05)--(04);
  \node[right=2.2cm of 05] (06) [hidadempty]{};
  \draw (05)--(06);
  \node[right=0.7cm of 06] (07) [hidad]{$v$};
  \draw (06)--(07);
 \end{tikzpicture}
\caption{
An example of~\eqref{eq:violates_aftertau} in Lemma~\ref{lm:errorevent}. Let the depth of confirmation $k=5$. Blocks $b$ and $e$ are on the same height. Block $b$ is the highest public chain at $s$ (block $g$'s parent is not yet public). The lead $L_s = 2$. Here $c$ and $f$ refer to the same block, i.e., $c=f$. At time $t_f$, block $g$ is $k$-confirmed by the highest chain $c$, as $h_c=h_g+k-1$. The highest disagreeing chain $z$ at $t_f$ is on height $h_z=h_g+2<h_c$. The (first) violation occurs at time 
    $d=t_v$, when chain $v$ is mined, which is credible as $h_v=h_c$. Height $h_g+1$ is balanced, with $X^g_{t_f}=1$.
    A single jumper is mined during $(\tau,d-\Delta]$. Three A-blocks are mined during $(s,d]$. 
    Therefore, we have $L_s+X_{t_f+\Delta}^g+A_{s,d}-J_{\tau,d-\Delta}=  5 = k$, i.e.,~\eqref{eq:violates_aftertau} holds in this case.}
   \label{fig:lemma10d}
\end{figure}

Based on Lemma~\ref{lm:errorevent}, 
we can characterize
the latency-security trade-off 
by examining the
competition between adversarial and honest nodes. 
While mining processes are completely 
memoryless, the height growths of honest and adversarial blocks are generally intertwined in a complex manner. In particular, when honest nodes have split views, the adversary may utilize some honest blocks to construct or maintain a competing chain, thereby perpetuating the split views among honest nodes. In general, the random variables in~\eqref{eq:violates_aftertau} are dependent. The challenge we must address is to decompose the event in Lemma~\ref{lm:errorevent} into separate, independent components. 

\subsection{A Sketch of the Statistical Analysis}
\label{subsec:sketch}

The remaining task is to analyze the probability of the event described by~\eqref{eq:violates_aftertau}. We need to understand the distribution of the individual elements as well as to account for their dependency.

Consider the balanced heights first.
As a convention, let
$(x)^+=\max\{0,x\}$.
 
\begin{theorem}
[Bounding the 
number of balanced heights]\label{thm:balanced heights}
\label{th:X}
    With respect to a given target block, the number of balanced heights up to the $k$-th candidacy, 
    denoted as $X^{(k)}$, 
    satisfies
    \begin{align}
        P(X^{(k)} \leq n) \geq F^{(k)}(n) \label{eq:P(Xk<n)_Fkn}
    \end{align}
    for every integer $n\ge0$, where $F^{(k)}$ is defined in~\eqref{eq:Fkn}.
As a special case, 
the total number of balanced heights defined in~\eqref{eq:defX} satisfies
\begin{align} \label{eq:PXFn}
  P(X 
  \leq n)  \geq F(n).
\end{align}
If~\eqref{eq:a>} holds, then $\delta=0$, so $F(n)=0$ and~\eqref{eq:PXFn} is vacuous.
\end{theorem}

Theorem~\ref{thm:balanced heights} is proved in Appendix~\ref{s:X}.  The key insight here is that the adversary can only take advantage of a bounded number of H-blocks to balance each other even in the long run.  In fact, with probability 1, one last balancing opportunity will be missed and the adversary is then behind future pacers permanently.  The final balancing candidate serves the role of the ``Nakamoto block'' in~\cite{dembo2020everything} and is a stronger notion than the ``convergence opportunity'' in e.g.,~\cite{pass2017rethinking}.  This observation underlies the safety guarantee of Theorem~\ref{th:Cher_upper}.

Let us define
\begin{align}
    Y^{(k)} = L_s - k + \sup_{d \geq \tau} \{A_{s,d}-J_{\tau,d-\Delta}\}.
    \label{eq:Yk}
\end{align}
While $Y^{(k)}$ and $X_{t_f+\Delta}^g$ are dependent, we can upper bound the probability of the event~\eqref{eq:violates_aftertau}, i.e., 
$X_{t_f+\Delta}^g+Y^{(k)}\ge0$, using a 
union bound:

\begin{lemma} \label{lm:x+y}
Let $\bar{F}_{X}(\cdot)$ denote the complementary cdf of the random variable $X$.
All integer-valued random variables $X$ and $Y$ satisfy
\begin{align} \label{lemma1}
      P(X + Y \geq 0)  \leq \inf _{n \in \mathbb{Z}} \{ \bar{F}_X(n-1) + \bar{F}_Y(-n)\} .
\end{align}
In addition, given two distributions of integer-valued random variables for which the right-hand side of~\eqref{lemma1} is strictly less than 1, there exists a joint distribution of $X$ and $Y$ that satisfy~\eqref{lemma1} with equality.
\end{lemma}

The proof of Lemma~\ref{lm:x+y} is relegated to Appendix~\ref{a:x+y}.

It suffices to upper bound the marginal complementary cdf's of 
$X_{t_f+\Delta}^g$ and $Y^{(k)}$.  Since all $X_{t_f+\Delta}^g$ balanced heights must be within the first $k$ heights above $h_b$, we have
\begin{align} \label{eq:PXn}
    P( X_{t_f+\Delta}^g > n-1 ) \le P( X^{(k)} > n-1 )
\end{align}
where the distribution of $X^{(k)}$ is given in~\eqref{eq:Fkn}.
Thus, we have 
\begin{align}
    P( Y^{(k)} & +X_{t_f+\Delta}^g \geq 0) \notag \\
    &\leq \inf _{n \in \mathbb{Z}} \{ \bar{F}_{X_{t_f+\Delta}^g}(n-1) + \bar{F}_{Y^{(k)}}(-n)\} \\
    & \leq  \inf _{n \in \mathbb{Z}} \{ \bar{F}_{X^{(k)}}(n-1) + \bar{F}_{Y^{(k)}}(-n)\} \label{eq:inf(FX+FY)}.
\end{align}

The remaining task is to understand the distribution of $Y^{(k)}$. A challenge here is that the jumper process depends on the H-block process. To decouple the dependence, we use the pacer process to lower bound the jumper process:

\begin{lemma}[There are more jumpers than pacers]\label{jp}
\begin{align}
    J_{s,t} \geq P_{s+\Delta,t},\quad \forall t > s .
\end{align}
\end{lemma}

\begin{proof}
    We first show that for every pacer block $p$, there must exist a jumper mined during $(t_p-\Delta,t_p]$.  This claim becomes trivial if block $p$ is the first 
    H-block mined on its height; otherwise, the first 
    H-block mined on $h_p$ must be mined strictly after $t_p-\Delta$.  Because pacers are mined at least $\Delta$ time apart, the lemma is established.
\end{proof}

By Lemma~\ref{jp}, we have
\begin{align}
    Y^{(k)}
    &\le L_s - k + \sup_{d \geq \tau} \{A_{s,d}-P_{\tau+\Delta,d-\Delta}\} \\
    &= L_s - k + \sup_{c \geq 0} \{A_{s,\tau+2\Delta+c}-P_{\tau+\Delta,\tau+\Delta+c}\} .
    \label{eq:YkLAP}
\end{align}
It is easy to see that the lead $L_s$, the A-block process from $s$ onward, and the pacer process from $\tau+\Delta$ onward are mutually independent. Also, the supremum in~\eqref{eq:YkLAP} can be viewed as the extremum of the difference of two renewal processes, whose distribution is readily available from Lemma~\ref{lm:LS_trans}.

Once the distribution of $Y^{(k)}$ is bounded, we can combine it with Theorem~\ref{th:X} to upper bound the probability of violation using~\eqref{eq:inf(FX+FY)}.
A detailed proof of Theorem~\ref{th:Cher_upper} is given in Appendix~\ref{s:Cher_upper}.

\subsection{Operational Implications}

A brief discussion on the operational implications of theorems presented in Sec. 2 is warranted. First, the requirement for ``a sufficiently large $s$'' in the theorems serves the purpose of ensuring that the maximum random pre-mining lead is in a steady state. This assumption is easily applicable in most practical scenarios.

More importantly, we define $s$ as ``given'' in the sense that it is not influenced by mining processes. Specifically, it is not selected by the adversary. If the adversary were to choose $s$, it could 
choose a target block one height above the highest public height once they have established a substantial lead (this might take an exceedingly long time, but it would eventually occur with a probability of 1). By doing so, a safety violation would be guaranteed. In this paper, we eliminate such {\it adaptive} attacks by assuming that 
the target block needs to 
extend a credible chain
that is beyond the adversary's control.

Operationally, we have the flexibility to select an arbitrary
time $s$. If an honest node sees the target block mined on top of a highest block in its view at time $s$, the node can be sure that the block is mined on $s$-credible chain, and hence satisfies the conditions in the theorems given in Sec.~\ref{s:theorem}. 


Following the preceding outline of the statistical analysis, we prove the main theorems in detail in the appendices. In Sec.~\ref{s:numerical}, we present the numerical results derived from these theorems. These results demonstrate the practicality of our latency-security trade-off bounds and their utility as guidance of parameter selection for Nakamoto-style protocols.

\section{Numerical Result}
\label{s:numerical}

In this section, the system parameters are chosen to represent two different settings in most cases.  The first setting is modeled after Bitcoin, where the total mining rate $
\lambda = 1/600$ blocks/second.
Bitcoin blocks' propagation delays vary and have been decreasing over the years.  As of November 2024, the 90th percentile of block propagation delay is about 4 seconds on average~\cite{bitcoinmoniter}.
We use $\Delta=10$ seconds as an upper bound.
The second setting is modeled after Ethereum Classic (ETC), where $\lambda 
= 1/13$ blocks/second and $\Delta = 2$ seconds. We also include results for Litecoin and Dogecoin parameters in one table.

\subsection{Security Latency Trade-off}

\begin{figure}
    \centering
    \includegraphics[width = .75\columnwidth]{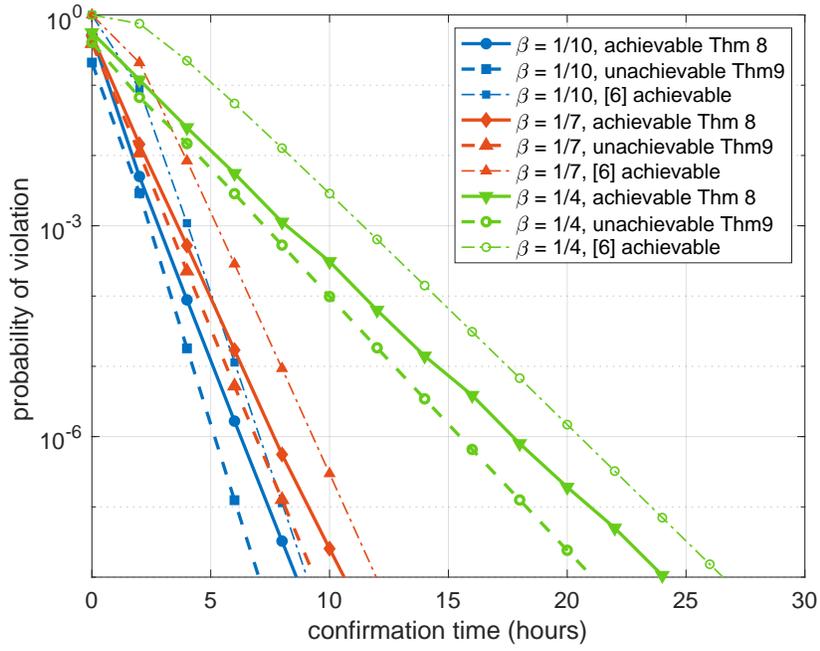}
    \caption{Bitcoin's safety versus confirmation time.}
    \label{fig:BTC}
\end{figure}

\begin{figure}
    \centering
    \includegraphics[width = .75\columnwidth]{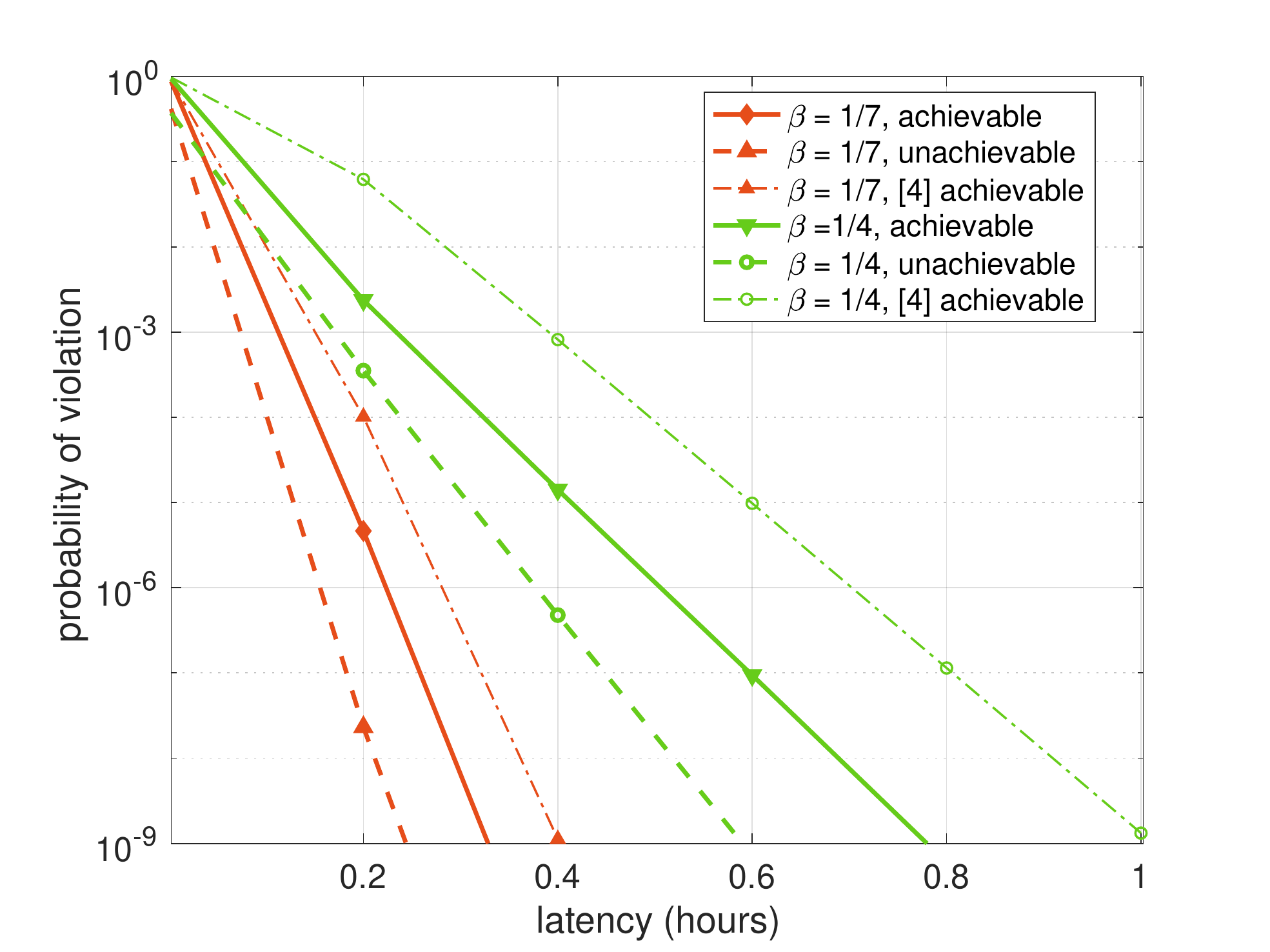}
    \caption{ETC's safety versus confirmation time.}
    \label{fig:ETH}
\end{figure}

In Fig.~\ref{fig:BTC},
we first show the achievable and unachievable probability of safety violation as a function of the confirmation time (in hours) for the Bitcoin configuration.  
For three different adversarial mining fractions ($\beta=1/10$, $1/7$, and $1/4$), we plot the upper and lower bounds
calculated using Theorems~\ref{th:upper_time} and~\ref{th:lower_time}, respectively.
Also included for comparison is an achievable trade-off from~\cite[Eqn.~(5)]{li2021close}.
In Fig.~\ref{fig:ETH}, we plot the latency-security trade-off of ETC for two different adversarial mining fractions ($\beta=1/7$ and $1/4$). In all cases, the 
new bounds approximately halve the latency gaps derived in~\cite{li2021close}.

In Fig.~\ref{fig:com-depth}, we show Bitcoin's safety bounds as a function of confirmation depth $k$ calculated using Theorems~\ref{th:Cher_upper},~\ref{th:upper_depth}, and~\ref{th:lower_depth},
respectively.  
Also included are the achievable and unachievable bounds from~\cite{doger2024refined}.
The new upper bound~\eqref{eq:upper_depth} of Theorem~\ref{th:upper_depth} yields 
tighter results for small confirmation depths. 
Similarly, ETC's safety bounds as a function of confirmation depth $k$ is plotted in Fig.~\ref{fig:com-depth_eth}.
For the ETC parameters, the achievable bounds in this paper 
outperform the achievable bound of~\cite{doger2024refined}.

\begin{figure}
    \centering
    \includegraphics[width = .75\columnwidth]{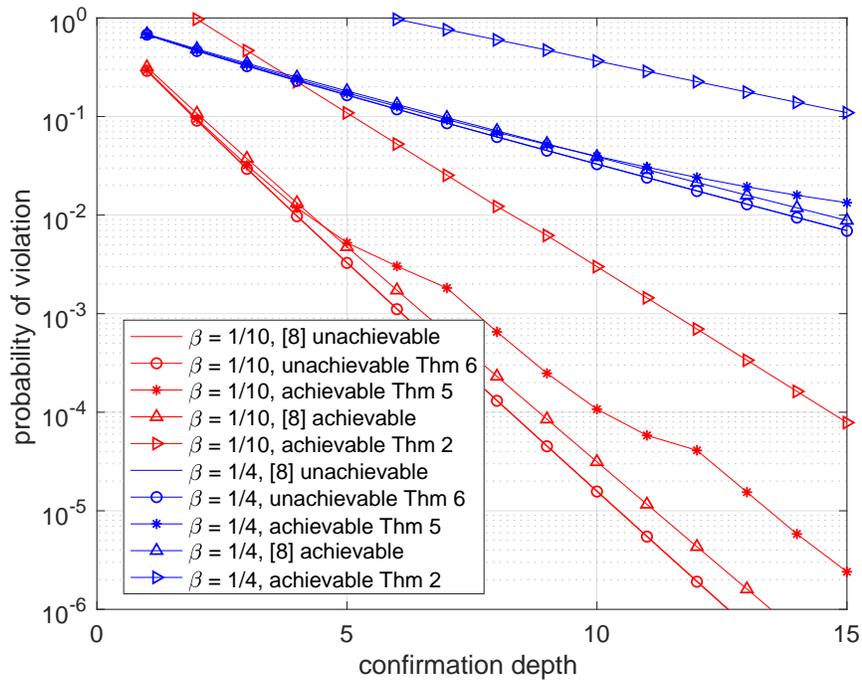}
    \caption{Bitcoin's safety versus confirmation depth.}
    \label{fig:com-depth}
\end{figure}

\begin{figure}
    \centering
    \includegraphics[width = .75\columnwidth]{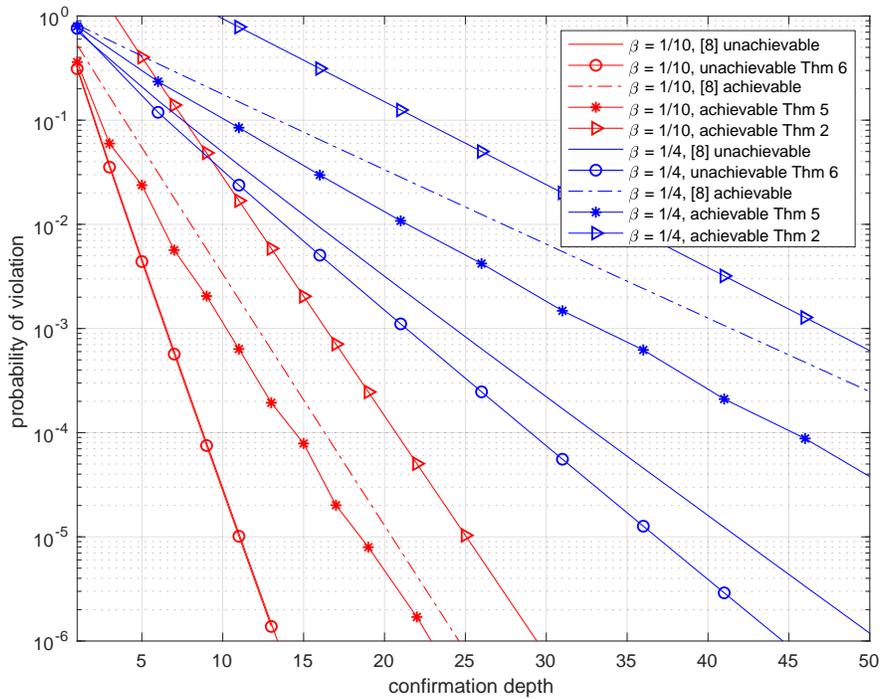}
    \caption{ETH's safety versus confirmation depth.}
    \label{fig:com-depth_eth}
\end{figure}

\begin{table}
\begin{center}
\begin{tabular}{|r|r|r|r|r|r|r|}
\hline
\rule{0pt}{8pt}
$\Delta$ & $0$ s & $2.5$ s & $10$ s & $40$ s & $160$ s & $640$ s  \\ 
\hline
$\beta=10 \%$ & $7$
&$7$ &$8$& $15$& $16$& $37$\\
\hline
$\beta=20 \%$ &$14$& $16$&$16$ &$19$& $28$& $65$\\
\hline
$\beta=30 \%$ &$34$& $37$ &$39$&$46$& $65$& $633$\\
\hline
$\beta=40 \%$ &$150$& $153$ &$156$&$277$& $552$& $\infty$\\
\hline
\end{tabular}    
\end{center}
\caption{The confirmation depth needed to guarantee $99.9\%$ safety for various block delays and adversarial fractions.}
\label{tb:delta_latency}
\end{table}

For a fixed $99.9\%$ safety guarantee, using the Bitcoin parameters, we compute the confirmation depth needed according to Theorem~\ref{th:upper_depth} for various block propagation delays and adversarial mining fractions.
Table~\ref{tb:delta_latency} shows that the confirmation depth rises slowly with the delay bound when the adversarial fraction is relatively low; it rises more quickly when the adversarial fraction is relatively high.

\begin{table}
\begin{center}
\begin{tabular}{|l|l|l|l|l|l|}
\hline
\rule{0pt}{8pt}
& $\bar{F}(0)$  & $\bar{F}(1)$ & $\bar{F}(2)$ & $\bar{F}(3)$&  $\bar{F}(4)$
\\ 
\hline
Bitcoin & $0.006$ & $0.0001$ & $2\times10^{-6}$ & $4\times10^{-8}$ & $8\times10^{-10}$
\\
\hline
ETC &$0.06$& $0.008$ &$0.001$& $2\times10^{-4}$& $3\times10^{-5}$
\\
\hline
\end{tabular}    
\end{center}
\caption{The complementary cdf, $\bar{F}(n)$.}
\label{tb:Xcdf_linear}
\end{table}

Table~\ref{tb:Xcdf_linear} shows the probability of having more than a given number of balanced heights with 25\% adversarial mining. For example, the probability of having more than four balanced heights is less than $10^{-9}$ for Bitcoin.
Hence the bound in Lemma~\ref{lm:x+y} is fairly tight for our purposes.
We note that for small confirmation depths, the probability of mining more adversarial blocks than pacers dominates the probability of balancing any height. As the confirmation depth increases, the latter probability begins to make a noticeable impact, hence the bump on the curves in Fig.~\ref{fig:com-depth}.

\begin{table}
\begin{center}
    \begin{tabular}{|l|r|r|r|}
    \hline
     &$\lambda^{-1}$ (s) & $\Delta$ (s) &$10^{-3}$  \\
    \hline
    Bitcoin \cite{guo2022bitcoin} & $600$ &10&$\le 25$\\
    \hline
    Bitcoin new &$600$ &10&$\le 24$\\
    \hline
    Litecoin \cite{guo2022bitcoin} &$150$ &10&$\le 36$\\
        \hline
    Litecoin new &$150$ &10&$\le 32$\\
        \hline
    Dogecoin \cite{guo2022bitcoin} &$60$ &10&$\le 85$\\
        \hline
    Dogecoin new &$60$ &10&$\le 52$\\
\hline
    ETC \cite{guo2022bitcoin} &$13$ &2&$\le 73$\\
        \hline
    ETC new &$13$ &2&$\le 47$\\
        \hline
    \end{tabular}
    \end{center}
    \caption{Bounds of confirmation depth (latency) to guarantee $99.9\%$ safety with $25\%$ adversarial mining power.}
    \label{tab:pow_compare}
\end{table}

Table~\ref{tab:pow_compare} presents the depth needed to guarantee a $99.9\%$
for Bitcoin, Litecoin, Dogecoin, and ETC, respectively.
The table shows that the new upper bound due to Theorem~\ref{th:upper_depth} is often substantially smaller than the previous
upper bound from~\cite{guo2022bitcoin}.


\subsection{Confirmation versus Fault Tolerance}

In this subsection, we plot the confirmation depth or time that guarantees a desired safety level as a function of the fraction of adversarial mining, while holding other parameters fixed. This also makes it easier to see the maximum fault tolerance different bounding methods can be applied to.  We base our analysis on the parameters of Ethereum Classic (ETC).

\begin{figure}
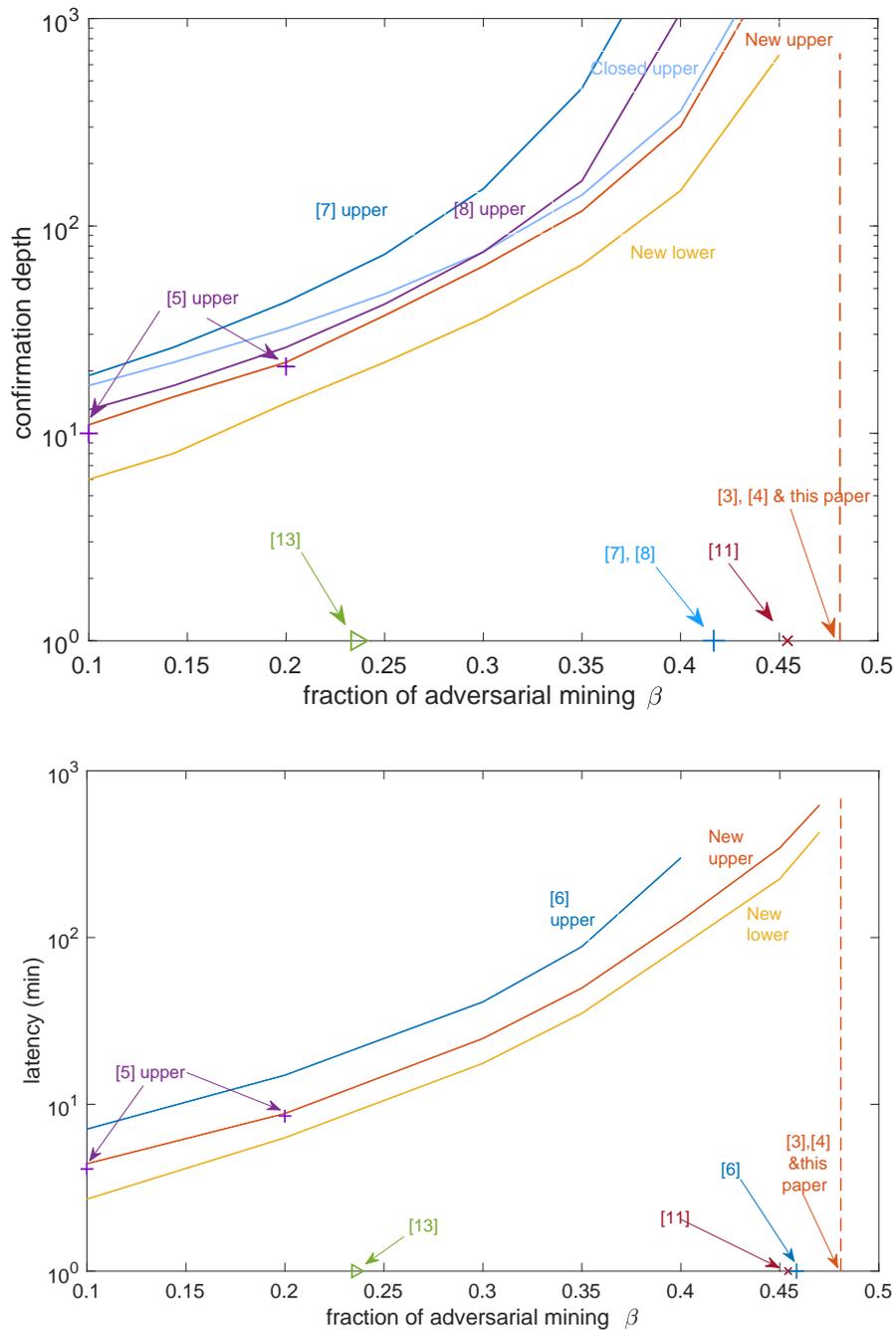

\begin{center}
    \includegraphics[width=.75\columnwidth]{fig/Fig7depth.pdf}\\
    \includegraphics[width=.75\columnwidth]{fig/latency-beta_ETH_0726.pdf} 
\end{center}
    \caption{Confirmation latency that guarantees 99.9\% safety for ETC assuming
    different fractions of adversarial mining.  Top: Confirmation in depth.  Bottom: Confirmation in time (minutes).}
    \label{fig:compare}
\end{figure}

Fig.~\ref{fig:compare} shows new and previous bounds on the confirmation depth/time that guarantees 99.9\% safety as a function of the adversarial fraction of mining power. 
The top graph uses confirmation by depth (Theorems~\ref{th:Cher_upper},~\ref{th:upper_depth}, and~\ref{th:lower_depth}), whereas the bottom graph uses confirmation by time (Theorems~\ref{th:upper_time} and~\ref{th:lower_time}).
We include two points reported in~\cite{gazi2022practical}, where the numerical evaluation becomes impractical for much higher adversarial mining rates.  Also included in the graph are achievable infinite-latency fault tolerances established in~\cite{li2021close},~\cite{guo2022bitcoin},~\cite{doger2024refined},~\cite{pass2017analysis},~\cite{garay2015bitcoin},
and the ultimate tolerance established in~\cite{dembo2020everything},~\cite{gazi2020tight}, and also in this paper.  The new upper and lower bounds are significantly closer than all previous bounds.  Even our new closed-form upper bound 
is highly competitive.



\subsection{Throughput-Latency Trade-Off}


In this subsection, we evaluate the best throughput we can guarantee for optimized parameters including block size, mining rate, and confirmation depth.  We fix the safety guarantee to be $99.9\%$.  We guarantee up to $\beta=40\%$ adversarial mining, which implies that $\lambda\Delta<5/6$ due to~\eqref{eq:fork<}. We assume the block propagation delay and the block size satisfy the affine relationship~\eqref{eq:throughput_delta}. Using empirically determined delay bounds of 10 seconds for 1.6 MB Bitcoin blocks and 2 seconds for 200 KB Bitcoin Cash blocks,\footnote{%
The 90th percentile of Bitcoin Cash propagation time has been measured as 1 second~\cite{fechner2022calibrating}.}
we infer that
$r=178.6$ KB/second
and 
$\nu = 0.9$ (seconds).


\begin{figure}
    \centering
\includegraphics[width = .75\columnwidth]{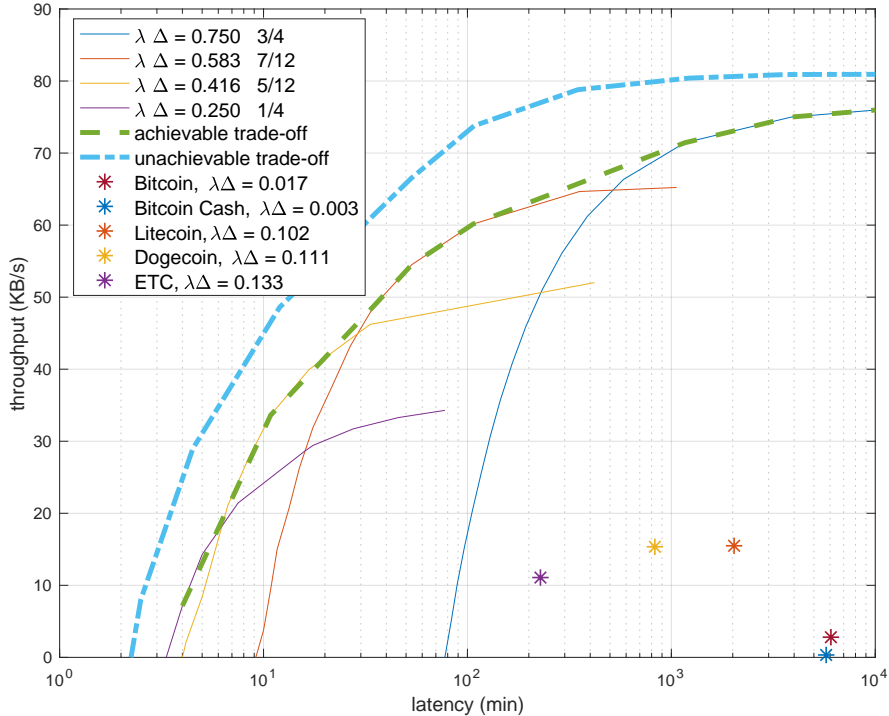}
    \caption{Achievable and unachievable throughput-latency trade-off and 5 existing currencies' throughput-latency.}
        \label{fig:thoughput}
\end{figure}

For every given expected confirmation latency $d$, we can solve the optimization problem~\eqref{eq:maximize} with $q=0.001$ and given natural fork rates $\lambda\Delta=3/4$, $7/12$, $5/12$, and $1/4$, respectively.  
For each given $\lambda\Delta$, we determine the confirmation depth $k$ needed to satisfy~\eqref{eq:p<=q} according to Theorem~\ref{th:Cher_upper}, which then allows us to determine the throughput for every choice of the average delay in time according to~\eqref{eq:maximize}.
Fig.~\ref{fig:thoughput} plots the throughput-latency trade-off  
for several different values of $\lambda \Delta$.  The upper envelop (in green dashed line) of all such curves represents the guaranteed throughput latency trade-off.
Also plotted in is the unachievable bound of Theorem~\ref{th:lower_depth}.
In particular, the maximum throughput is approximately $81$ KB/second due to~\eqref{eq:throughput_r}.
We observe that the choice of $\lambda\Delta=7/12$ yields a large fraction of the maximum possible throughput for practical latencies of $30$ minutes or more.

\subsection{Guidance for Selecting Parameters}

Bitcoin's block size ($1.6$ MB) and mining rate ($1$ every $10$ minutes) have proven to be safe choices, with a maximum throughput of about $2.8$ KB/s.  The block size has long been an issue of debate.  Using a larger block size while holding other parameters unchanged may increase the throughput (and perhaps lower transaction fees), but at the same time degrade the safety
due to larger propagation delays and increased forking (and perhaps also tendency of mining centralization). 
As another important factor, increasing the mining rate increases natural forks, but also reduces the expected time to the same confirmation depth.

The theorems in this paper provide the most accurate guidance to date for selecting parameters for desired system-level performance. 
For example, if the block size is increased to $2$ MB, the propagation upper bound will be around $\Delta=12.5$ seconds according to~\eqref{eq:throughput_delta}.  If the total mining rate is increased to $\lambda=1/30$ block/seconds,
the throughput increases to over $40$ KB/s according to~\eqref{eq:throughput}.
At $\lambda \Delta = 5/12$, the ultimate fault tolerance is 
$44.44\%$ adversarial mining; in this case, at $30$\% adversarial mining, the latency for achieving $10^{-6}$ safety is less than $3.5$ hours, while original Bitcoin takes more than $18$ hours.  Thus the throughput, latency, and safety are simultaneously improved for a minor sacrifice in the fault tolerance.  In fact, $44.44\%$ adversarial mining is practically not tolerable using the original Bitcoin parameters due to the excessively long delays of thousands of block confirmations. 

\section{Conclusion}

In this paper, we have developed a concise mathematical model for blockchain systems based on the Nakamoto consensus. We have obtained a new set of closed-form latency-security bounds, 
which are the first such bounds to achieve the ultimate fault tolerance.
These bounds are not only easily computed but also demonstrate superior numerical performance across a broad spectrum of parameter configurations. Our analysis not only provides an intuitive proof of the ultimate fault tolerance but also offers new insights for the practical application the Nakamoto consensus.

Existing safety guarantees for longest-chain fork-choice
protocols have primarily been
established through 
``block counting'' methods, which calculates, 
the probability that 
one species of blocks is
mined more than another species during some periods longer than the confirmation latency.  Such analyses essentially empower the adversary to use non-causal knowledge of 
future blocks (up until confirmation) and advance them to build a winning chain. 
In contrast, 
the production of a balanced height depends on 
blocks 
mined within an opportunity window of $\Delta$ seconds; future blocks cannot be used to create balancing candidates or balanced heights.  This requirement eliminates certain 
non-causal attack strategies, 
resulting in improvements and achieving the ultimate
fault tolerance.

In today's Bitcoin and future open blockchain networks, miners play 
a crucial role in investing 
in infrastructure to mitigate block propagation delays.  Additionally, constant monitoring of delays and vigilant of observation of potential network partitions are crucial aspects of their operations.  During periods of outage or network disruption, an accurate estimation of delays enables miners and users to adapt their confirmation latency to stay safe.  Subsequently, as delays normalize, forks are easily resolved, and nodes can revert to utilizing smaller confirmation latency specifically designed for standard conditions. The theorems presented in this paper serve as useful tools, empowering all stakeholders to make informed economic decisions 
regarding their operations.


\section*{Acknowledgement}

D.~Guo thanks Dr.\ David Tse for stimulating discussions that sparked the conceptualization of some initial ideas in the paper.  This work was supported in part by the National Science Foundation (grant No.~2132700).

\appendices

\section{Proof of Theorem~\ref{th:unsafe}}
\label{a:unsafe}

\begin{figure*}
    \centering
    \includegraphics[width = 0.88\textwidth]{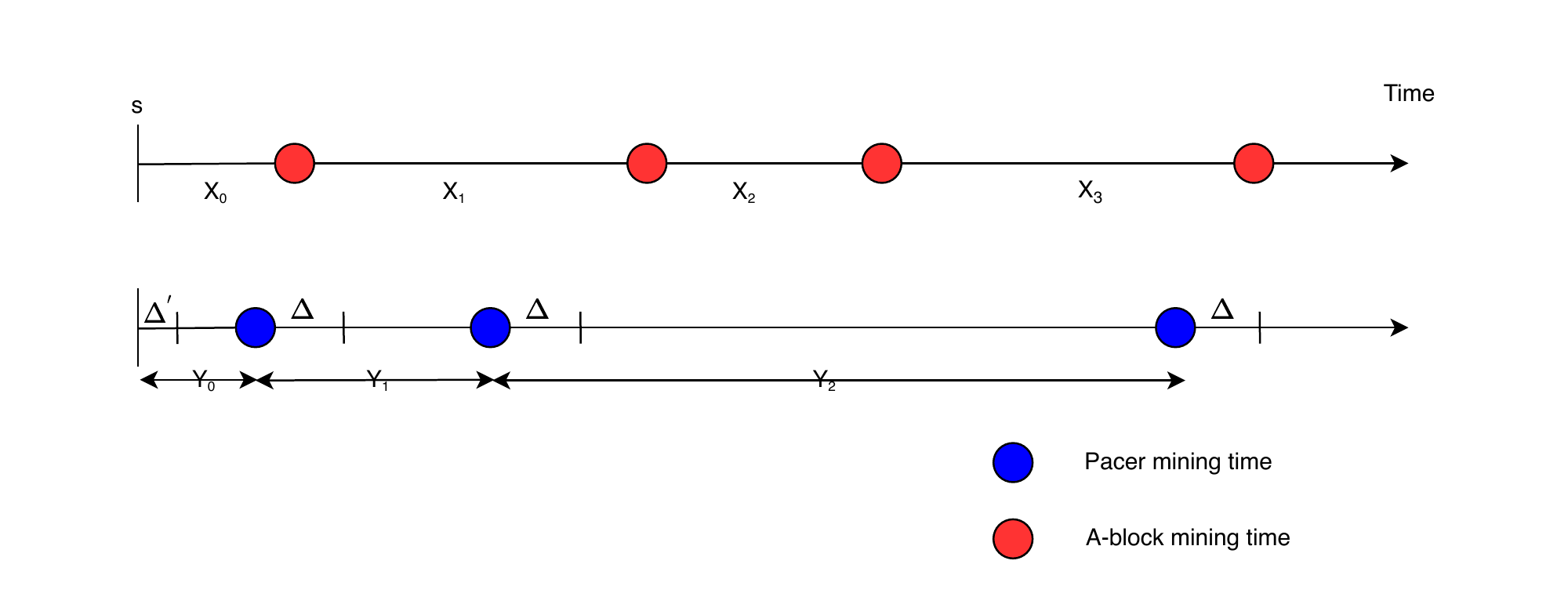}
    \caption{Illustration of mining times of pacers and A-blocks after time $s$.}
    \label{fig:pacer_ad}
\end{figure*}

Under a confirmation rule of latency $t$, 
a target block is 
mined by time $s$.
Within this proof, we count block arrivals from time $s$, i.e., the first A-block is the first A-block mined after $s$, and the first pacer is the first pacer mined after $s$.
Let $X_0$ denote the time from $s$ to the mining time of the first A-block.
For $i = 1,2,\dots$,
let $X_i$ denote the inter-arrival time between the $i$-th A-block and the $(i+1)$-st A-block.
Then, $X_0,X_1,\dots$ are independent and identically distributed (i.i.d.) exponential random variables with mean $\frac{1}{\mya}$.
Let $Y_0$ denote the time from $s$ to the mining time of the first pacer.
For $i = 1,2,\dots$, let $Y_i$ denote the inter-arrival time between the $i$-th pacer and the $(i+1)$-st pacer.
It is easy to see that $Y_1,Y_2,\dots$ are i.i.d.\ random variables with mean $\frac1{\myh}+\Delta$, and $Y_0$ has mean less than $\frac1{\myh}+\Delta$.
An illustration of these random variables is given in Fig.~\ref{fig:pacer_ad}.

Let $E_j$ denote the event that the $j$-th A-block arrives before the $j$-th pacer but after time $s+t$.
The probability of safety violation is at least
\begin{align}
    P(E_j)
    &= P \Bigg(\sum_{i = 0}^{j} X_i < \sum_{i = 0}^{j} Y_i,\sum_{i = 0}^{j} X_i > t \Bigg)\\
    &\geq 1-P\Bigg(\sum_{i = 0}^{j} (X_i-Y_i)\geq 0\Bigg)
    - P\Bigg(\sum_{i = 0}^{j} X_i \leq t \Bigg) \\
    &\geq 1-P\Bigg(\sum_{i = 1}^{j} (X_i-Y_i)\geq -n \Bigg) - P(X_0-Y_0\ge n)     -P\Bigg(\sum_{i = 0}^{j} X_i \leq t \Bigg) \label{eq:pE}
\end{align}
where~\eqref{eq:pE} holds for every integer $n$.  For every $\epsilon>0$, there exists $n$ such that $P(X_0-Y_0\ge n)<\epsilon/3$.  For sufficiently large $j$, we also have
$P(\sum_{i = 0}^{j} X_i \leq t ) < \epsilon/3$.
By~\eqref{eq:a>}, we have $E[X_1-Y_1] < 0$.
Thus, by the law of large numbers, 
\begin{align}
    P\Bigg(\sum_{i = 1}^{j} (X_i-Y_i)\geq -n \Bigg) < \frac\epsilon3.
\end{align}
We conclude that $P(E_j)\rightarrow  1$ as $j \rightarrow \infty$.

\section{Proof of Theorem \ref{th:X}}
\label{s:X}


\subsection{Critical Events}


\begin{figure}
    \centering
    \includegraphics[width=.75\columnwidth]{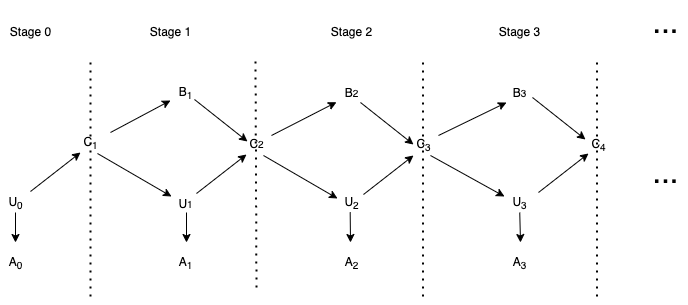}
    \caption{Illustration of all possible state transitions over time.}
    \label{fig:block}
\end{figure}

Let a 
target block be given.
As a convention, we refer to the genesis block as the 0-th balancing candidate.
We describe several 
types of events as follows:
For every $i \in \{0,1,\cdots\}$, let
\begin{itemize}
    \item $C_i$ denotes the event that $i$ or more candidates are eventually mined;
    \item $B_i$ denotes the event that the $i$-th candidate is mined and subsequently balanced;
    \item $U_i$ denotes the event that the $i$-th candidate is mined but is not balanced by the end of its candidacy;
\end{itemize}
Evidently, $B_i \cup U_i = C_i$.
Also, $C_0, C_1, \dots$ is decreasing in the sense that $C_{i+1}\subset C_i$. By definition, $C_i$ is prerequisite for $C_{i+1}$. E.g., ``3 or more" means $\{3,4,...\}$ and is included in ``2 or more", which means $\{2,3,4,...\}$. For convenience, we define
\begin{align}
    A_i=C_{i} \cap \bar{C}_{i+1}
\end{align}
which denotes the event that the $i$-th candidate is the last candidate ever mined.

As time progresses, the events are seen to occur in natural order as illustrated in Fig.~\ref{fig:block}.  The process begins from $U_0$ as the genesis block, which is also the 0-th candidate, cannot be balanced.  If a first candidate is ever mined, the process transitions to $C_1$, otherwise it terminates at $A_0$.
In particular, for every $i=1,2,\dots$, if $C_i$ ever occurs,
then either $B_i$ or $U_i$ occurs.  If $B_i$ occurs, then with probability 1 the next candidate will be mined,
i.e., $C_{i+1}$ must occur.  If $U_i$ occurs instead, then there is some chance $C_{i+1}$ occurs as well as some chance that no more candidates are mined, so that the process terminates at $A_i$.


Conditioned on the $k$-th candidate is mined, whose candidacy is no longer than $\Delta$, the probability that it is also balanced is no higher than the
probability that an honest block is mined over a period of $\Delta$.  Hence
\begin{align}\label{lowerbound_CB}
    P(B_{k}|C_{k}) \leq \epsilon
\end{align}
where $\epsilon$ is given by~\eqref{eq:epsilon}.

If $U_k$ occurs, then at the end of the $k$-th candidacy, the candidate is not balanced, so the candidate becomes public and is strictly higher than all disagreeing blocks.  As such, all honest mining power will extend chains that include the target transaction and are higher than the $k$-th candidate. 
This implies that the adversary must extend a shorter chain to catch up with the highest public honest block in order for the next candidate to be mined.
In Appendix~\ref{a:honest_escape}, we show that the chance that no additional candidate will ever be mined is lower bounded as
\begin{align}\label{eq:lowerbound_UA}
    P(A_{k}|U_{k}) \ge \delta
\end{align}
where $\delta$ is given by~\eqref{eq:delta}.
The lower bound~\eqref{eq:lowerbound_UA} is positive only when~\eqref{eq:a<} holds.

\subsection{Proof of~\eqref{eq:lowerbound_UA}}
\label{a:honest_escape}
Here we wish to lower bound $P(A_k|U_k)$, which is the probability that no more candidate is ever mined after that the $k$-th candidate is not balanced by the end of its candidacy.  We refer to this as the honest miners' escape.  
There is no chance to escape if $a>h/(1+h\Delta)$.  Otherwise, an escape occurs if (but not only if) the following events occur:
1) during the $k$-th candidacy, no A-block is mined, so that when the candidate becomes public, it is higher than all A-blocks;
2) from the end of the candidacy, say time $r$, the number of A-blocks mined during $(r,r+t]$ is no more than the number of jumpers mined during $(r,r+t-\Delta]$ for every $t\geq 0$, so that the adversary never produces a disagreement about the target transaction after $r$.

Since the duration of the candidacy is at most $\Delta$, the probability of the first event is evidently lower bounded by $e^{-a\Delta}$.
Conditioned on the first event, once the $k$-th candidacy ends, the next H-block is a jumper, which arrives after an exponential time.  The chance of an escape is minimized by delaying the propagation of all H-blocks maximally from $r$ onward.  In this case, the jumper process $J_{r,r+t-\Delta}$ is statistically identically to the pacer process $P_{\Delta,t}$.  Conditioned on the first event, the second event's probability is thus lower bounded as
\begin{align}
    &P\left( \cap_{t\ge 0}  \{ A_{r,r+t} \le P_{\Delta,t} \}\right)
    =
    P\left( \cap_{t\ge0} \{ A_{0,t} \le P_{0,t} \} \right)    \\
    &=
    P\left( \max_{t\ge0} \{ A_{0,t} - P_{0,t} \} = 0 \right) \\
    &=
    P( M=0 ) .
\end{align}
Based on the preceding bounds and the pmf of $M$ given by~\eqref{eq:e(i)}, we have
\begin{align}
    P( A_k | U_k )
    &\ge e^{-a\Delta} P(M=0) \\
    &= \delta.
\end{align}

\subsection{Partially Markov Processes}


Consider a sample space in which each outcome corresponds to a (sample) path on the directional graph depicted by Fig.~\ref{fig:block} from $U_0$ to some $A$ vertex.  If the path ends at $A_k$, the outcome can be represented by a binary string 
$0\, b_1 \dots b_{k-1} 0$,
where $b_i=1$ if $B_i$ is on the path (in stage $i$).  The corresponding number of balanced heights is thus $X=b_1+\dots+b_{k-1}$.

Naturally, we let each vertex in the graph denote the event which consists of all paths that go through the vertex.
Under condition~\eqref{eq:a<}, with probability 1 some $A_k$ event terminates the sample path.  Thus every adversarial strategy induces a pmf on the (discrete) sample space, which must satisfy~\eqref{lowerbound_CB} and~\eqref{eq:lowerbound_UA}.  Under such a pmf, the transitions in the graph are not Markov in general.  In particular, we cannot say the event $B_i$ is independent of $B_{i-1}$ conditioned on $C_i$.

Let $P^{0}(\cdot)$ stand for the special pmf under which all state transitions on the directed graph are Markov with the following transition probabilities:  For every $k \in \{0,1,...\}$,
\begin{align}
    P^{0}(B_k | C_k)
    &= 1 - P^{0}(U_k | C_k)
    = \epsilon,
    \label{markov_const_s}\\
    P^0 (A_{k} | U_k)
    &= 1 - P^{0}(C_{k+1} | U_k) 
    = \delta,
    \label{markov_const_e}
\end{align}
where $\epsilon$ and $\delta$ are given by~\eqref{eq:epsilon} and~\eqref{eq:delta}, respectively.
These transition probabilities satisfy the prior constraints~\eqref{lowerbound_CB} and~\eqref{eq:lowerbound_UA} with equality under condition~\eqref{eq:a<}.

We next prove that every pmf $P$ which satisfies~\eqref{lowerbound_CB} and~\eqref{eq:lowerbound_UA} must also satisfy
\begin{align}\label{eq:target}
       P(X \le n) \ge P^0 (X \le n) .
\end{align}
Towards this end, we define a sequence of pmfs, $P^1, P^2,\dots$, as follows:  Let $0\,b_1\dots b_i*$ denote the event consisting of all sample paths with the prefix $0\,b_1\dots b_i$.
Let the probability of a sample path under $P^m$ be defined as
\begin{align} \label{eq:Pm1}
    P^m( \{ 0\,b_1\dots b_{k-1}0 \} )
    = P( \{ 0\,b_1\dots b_{k-1}0 \} )
\end{align}
for 
$k=1,\dots,m-1$ and as
\begin{align} \label{eq:Pm2}
\begin{split}
     P^m( \{ 0\,b_1\dots b_{k-1}0 \} ) 
    = P( C_m,\, 0\,b_1\dots b_{m-1}* ) \times P^0( \{ 0\,b_1\dots b_{k-1}0 \} | C_m,\, 0\,b_1\dots b_{m-1}* )
 \end{split}
\end{align}
for 
$k=m, m+1, \dots$. 
The outcome $0\,b_1\dots b_{k-1}0$ with $k\ge m$ 
implies that $C_m$ occurs.
It is straightforward to verify that~\eqref{eq:Pm1},~\eqref{eq:Pm2} defines a pmf.
Evidently, $P^m$ is partially Markov in the sense that it is consistent with $P$ up until stage $m-1$, after which $P^m$ becomes Markov with the same transition probabilities as $P^0$.  In particular, conditioned on $C_m$, the balanced stages before stage $m$ and the balanced stages afterward are independent.

For every pair of natural numbers $l$ and $j$, we define the following pair of random variables as functions of the sample path $0\,b_1\dots b_{l-1}0$:
    \begin{align}
        L_j &= b_1 + \dots + b_{\min(l,j)-1} \\
        R_j &= b_j + \dots + b_{l-1} ,
    \end{align}
with the convention that $L_0 =0$.  Evidently, with $1_{B_0}=0$
\begin{align}
    X = L_m + 1_{B_m}+ R_{m+1} \label{eq: X=L+B+R}
\end{align}
holds for every natural number $m$.

In Appendices~\ref{a:part1} and~\ref{a:part2}, we show that for every $m=0,1,\dots$ and every 
$n = 0,1,\dots$,
\begin{align}
       P^{m+1} (X \le n) \ge P^{m} (X \le n) \label{part_1}
\end{align}
and
\begin{align}
    \lim_{m\to\infty} P^m(X\le n) = P(X\le n) \label{part_2}
\end{align}
hold.
As the limit of an increasing sequence, $P(X\ge n)$ must satisfy~\eqref{eq:target}.

For $n=0,1,\dots$, let us denote
\begin{align}
    f_n &= P^0(X \geq n|C_1) \\
    &=P^0(R_1 \geq n|C_1).
\end{align}
Then we have the following recursion
\begin{align}
    f_n
    =& P^0(R_1\geq n, U_1|C_1) + P^0(R_1\geq n, B_1|C_1) \\
    =& P^0(R_2\geq n | C_2) P^0(C_2|U_1) P^0(U_1|C_1)\notag \\
    &+ P^0(R_2\geq n-1 | B_1) P^0( B_1|C_1) \\
    =& 
    (1-\delta) (1-\epsilon) f_n
    +
    \epsilon f_{n-1}.
\end{align}
Solving the recursion with the fact that $f_0=1$ yields
\begin{align}
     f_n = \left( \frac{\epsilon}{\epsilon+\delta-\epsilon\delta} \right)^n .
\end{align}
It is easy to verify that
\begin{align}
    P^0(X\le n)
    &= 1 - P^0(R_1\geq n+1|C_1) P^0(C_1)\\
    &= F(n)  \label{eq:P0Fn}
\end{align}
where $F(\cdot)$ is defined in~\eqref{eq:Fn}.

With confirmation by depth, there can be at most $k$ balanced heights before violation.  
In the remainder of this section, we are concerned with the number of balanced candidates up to the $k$-th candidacy, denoted as $X^{(k)}$. We first establish the following result:
\begin{align}
    P(X^{(k)} \leq n)
    \geq P^0(X^{(k)} \leq n) \label{eq:P-P0}
\end{align}
where~\eqref{eq:P-P0} can be proved as follows.
Since we only consider the first $k$ candidates, which means only the first $k$ stages are relevant, we have
\[
    P(X^{(k)} \leq n) = P^{k+1}(X^{(k)} \leq n).
\]
Furthermore, we can prove that for every $m = 0, 1, \dots, k$ and $n = 0, 1, \dots, k$,
\begin{align}
    P^{m+1}(X^{(k)} \leq n) \geq P^{m}(X^{(k)} \leq n).
\end{align}
For $X^{(k)}$, we define $L'_j$ and $R'_j$ as functions of the sample path $0 \, b_1 \dots b_{l-1} 0$ for natural number $l$:
\begin{align}
    L'_j &= b_1 + \dots + b_{\min(j,l) - 1}, \quad j = 0, 1, \dots, k, \\
    R'_j &= b_j + \dots + b_{\min\{k,l\}}, \quad j = 1, 2, \dots, k+1,
\end{align}
with boundary conditions $L'_0 = 0$ and $R'_{k+1} = 0$.
Thus, we can express $X^{(k)}$ as:
\begin{align}
    X^{(k)} = L'_m + 1_{B_m} + R'_{m+1}, \quad \text{for } m \leq k.
\end{align}
The remainder of the proof follows the same steps as in the proof of~\eqref{part_1} in Appendix~\ref{a:part1}, the details of which is omitted here.

Using~\eqref{eq:P-P0}, 
we can lower bound the cdf of $X^{(k)}$: 
\begin{align}
    P(X^{(k)} \leq n)
    &= P^0(X^{(k)} \leq n,X \geq n+1)+P^0(X \leq n) \label{eq:P0+P0}\\
    &= P^0(X^{(k)} \leq n,X \geq n+1, C_{k+1})+P^0(X \leq n) \label{eq:P0c+P0}
\end{align}
where~\eqref{eq:P0c+P0} holds because $X^{(k)} \leq n$ and $X \geq n+1$ implies the occurrence of $C_{k+1}$.
Since we already know $P^0(X \leq n)$ by~\eqref{eq:P0Fn}, what left is to calculate $P^0(X^{(k)} \leq n,X \geq n+1, C_{k+1})$.
By law of total probability, we have
\begin{align}
\begin{split}
   P^0(X^{(k)} \leq n,X \geq n+1, C_{k+1}) = \sum_{i=0}^{n} P^0(X^{(k)}=i,C_{k+1})P^0(X-X^{(k)} \geq n+1-i|C_{k+1}) . \label{eq:sum_Xk=i}
\end{split}
\end{align}
Due to Markov property on $P^0$, $X-X^{(k)} \geq n+1-i$ conditioning on $C_{k+1}$ is identically distributed as $X  \geq n+1-i$ conditioning on $C_{1}$.  $P^0(X^{(k)}=i,C_{k+1})$ equals the sum of probabilities of all possible paths from $U_0$ to $C_{k+1}$ which traverse $i$ of the B states. For every positive integer $m$, conditioned on $C_m$, the probability of going through $B_m$ to arrive at $C_{m+1}$ is $\epsilon$, while the probability of not going through $B_m$ to arrive at $C_{m+1}$ is $(1-\epsilon)(1-\delta)$. Thus, the number of B states visited between $C_1$ and $C_{k+1}$ is a binomial random variable.  We have
\begin{align}
    &P^0(X^{(k)} \leq n,X \geq n+1, C_{k+1}) \notag \\
    &=\sum_{i=0}^{n} \binom{k}{i}\epsilon^i((1-\epsilon)(1-\delta))^{k-i}
    (1-\delta) \Big(\frac{\epsilon}{\epsilon+\delta-\epsilon\delta}\Big)^{n-i+1}\label{eq:to-parameter}\\
    &= (1-\delta)\Big(\frac{\epsilon}{\epsilon+\delta-\epsilon\delta}\Big)^{n-k+1} \sum_{i=0}^{n} \binom{k}{i} \epsilon^i \Big(\frac{\epsilon(1-\epsilon)(1-\delta)}{\epsilon+\delta-\epsilon\delta}\Big)^{k-i} \\
    &=  (1-\delta)\Big(\frac{\epsilon}{\epsilon+\delta-\epsilon\delta}\Big)^{n-k+1}\Big(\frac{\epsilon}{\epsilon+\delta-\epsilon\delta}\Big)^{k} \sum_{i=0}^{n} \binom{k}{i} (\epsilon+\delta-\epsilon\delta)^i ((1-\epsilon)(1-\delta))^{k-i}
    \\
    &=  (1-\delta)\Big(\frac{\epsilon}{\epsilon+\delta-\epsilon\delta}\Big)^{n+1}F_b(n;k, (\epsilon+\delta-\epsilon\delta)), \label{eq:P0Xk<}
\end{align}
    where $F_b(\cdot;k,p) $ denotes the cdf of a binomial distribution with $k$ independent trials and success probability $p$. 
Plugging~\eqref{eq:P0Xk<} back into~\eqref{eq:sum_Xk=i} completes the proof of~\eqref{eq:Fkn} and also Theorem~\ref{thm:balanced heights}.

\subsection{Proof of~\eqref{part_1}}
\label{a:part1}

We use $C_{m,i}$ to denote the intersection of events $C_m$ and $L_m=i$.  Then
\begin{align}
   P^{m+1}(R_{m+1} < c \,|\, C_{m+1,i}, B_m )&=
    P^0(R_{m+1} < c \,|\, C_{m+1} ) \label{eq:R_0}\\
    &=
    P^{m+1}(R_{m+1} < c \,|\, C_{m+1,i}, U_m) .
\end{align}
Also, partially Markov process with measure $P^m(\cdot)$ becomes memoryless once event $C_m$ occurs:
\begin{align}
    P^{m}(A_m |C_{m,i}) &= P^{0} (A_m |C_{m,i})\label{eq:A_m0}\\
    P^{m}(B_m|C_{m,i}) &= P^{0} (B_m|C_{m,i})\label{eq:B_m0}\\
    P^{m}(C_{m,i}) &= P(C_{m,i})\label{eq:m_P}\\
    &=P^{m+1}(C_{m,i}).\label{eq:m+1_P}
\end{align}

Consider the right-hand side of~\eqref{part_1}:
\begin{align}
    &P^m(X \le a)
    = P^m(X \le a , C_m^c) + P^m(X\le a, C_m)\label{X_LR}\\
    &= P^m(L_m \le a , C_m^c) + \sum^a_{i=0}  P^m (C_{m,i})P^m ( X\le a| C_{m,i} )
    \label{condition_CL}\\
    &= P(L_m \le a , C_m^c) + \sum^a_{i=0}  P (C_{m,i}) P^m ( X\le a| C_{m,i} )
    \label{eq:m+1_m}
\end{align}
where~\eqref{eq:m+1_m} is due to the fact that $P^m$ is identical to $P$ up until stage $m-1$ as in~\eqref{eq:m_P}.
Similarly, by~\eqref{eq:m+1_P}
\begin{align}
\begin{split}
    P^{m+1}(X \le a)
    = P(L_m \le a , C_m^c) + \sum^a_{i=0} P (C_{m,i}) P^{m+1} ( X\le a| C_{m,i} ) .
    \label{m+1_m}    
\end{split}
\end{align}
In order to establish~\eqref{part_1}, it suffices to show
\begin{align}
    P^{m+1} ( X\le a| C_{m,i} )
    \ge
    P^m ( X\le a| C_{m,i} )
\end{align}
holds for all $i\in\{0,\dots,a\}$.

Since
$
    C_m
    =  B_m \cup(U_m, C_{m+1}) \cup  A_m,
$we have
\begin{align}
    &P^{m+1} ( X\le a| C_{m,i} ) \notag \\=&
    P^{m+1} ( X\le a, B_m | C_{m,i} )
    + 
    P^{m+1} ( X\le a, A_m  | C_{m,i} )
    +  
    P^{m+1} ( X\le a, U_m, C_{m+1} | C_{m,i} )
    \\
    =&
    P^{m+1} ( R_{m+1}\le a-i-1, B_m | C_{m,i} )
    +P^{m+1} ( A_m  | C_{m,i} )+  P^{m+1} (R_{m+1}\le a-i, U_m, C_{m+1} | C_{m,i} ) \label{eq:X_R}\\
    =&
    P^{m+1} (R_{m+1}\le a-i-1 | C_{m+1} )
    P^{m+1} ( B_m, C_{m+1} | C_{m,i} ) + P^{m+1} (R_{m+1}\le a-i | C_{m+1} ) 
    P^{m+1} ( U_m, C_{m+1} | C_{m,i} )\notag \\
    &+ P^{m+1} ( A_m  | C_{m,i} )
    \\
    =&
    P^{m+1} ( B_m, C_{m+1} | C_{m,i} )
    x_1
    +
    P^{m+1} ( U_m, C_{m+1} | C_{m,i} ) 
    x_2 
    + P^{m+1} ( A_m  | C_{m,i} )\label{eq:x123}
\end{align}
where
\begin{align}
    x_1 &= P^0 ( R_{m+1}\le a-i-1 | C_{m+1} ) \\
    x_2 &= P^0 ( R_{m+1}\le a-i | C_{m+1} ) ,
\end{align}
and~\eqref{eq:X_R} is due to~\eqref{eq: X=L+B+R},\eqref{eq:x123} is due to~\eqref{eq:R_0}.
Similarly, with~\eqref{eq:A_m0} and~\eqref{eq:B_m0}, we have
\begin{align}
\begin{split}
    P^{m} & ( X\le a| C_{m,i} )
    =
    P^{0} ( B_m, C_{m+1} | C_{m,i} )
    x_1 \\
    &
    +
    P^{0} ( U_m, C_{m+1} | C_{m,i} ) 
    x_2+       P^{0} ( A_m  | C_{m,i} ).
\end{split}
\end{align}
Due to~\eqref{lowerbound_CB} and~\eqref{eq:lowerbound_UA},
\begin{align}
    P^{0} ( A_m  | C_{m,i} ) &\leq P^{m+1} ( A_m  | C_{m,i} )\\
     P^{0} ( B_m, C_{m+1} | C_{m,i} ) &\geq  P^{m+1} ( B_m, C_{m+1} | C_{m,i} ).
\end{align}

Now we introduce the following claim to assist the proof, which enables us to compare differently weighted sums of  three numbers.
\begin{lemma} \label{lm:lambda}
    If $\lambda_1 + \lambda_2 +\lambda_3 = 1$, $\lambda_1' + \lambda_2' +\lambda_3' = 1$ and $\lambda_1 \geq \lambda_1' $, $\lambda_3 \leq \lambda_3'$, $\lambda_i \text{ and } \lambda_i' \geq 0$ for $i =1,2,3$.
    Then all real numbers $x_{1} \leq x_{2} \leq 1$ satisfy
    \begin{align}
        \lambda_1 x_1 + \lambda_2 x_2+
        \lambda_3 
        \leq
        \lambda_1' x_1 + \lambda_2' x_2+
        \lambda_3'.
    \end{align}
\end{lemma}

\begin{proof}
    The proof is straightforward:
      \begin{align}
        \lambda_1 x_1  + \lambda_2 x_2 + \lambda_3 
        &= x_2+\lambda_1(x_1-x_2)+\lambda_3(1-x_2)\\
       & \leq 
         x_2+\lambda_1'(x_1-x_2)+\lambda_3'(1-x_2)\\
         &=\lambda_1' x_1 + \lambda_2' x_2+
        \lambda_3' .
      \end{align}
    \end{proof}

Lemma~\ref{lm:lambda} implies
\begin{align}
    P^{m+1}(X\leq a|C_{m,i})&\geq P^{0}(B_m, C_{m+1}|C_{m,i})x_1+P^{0}( U_{m}, C_{m+1}|C_{m,i})x_2+ P^{0} (A_m |C_{m,i}) \\
    &= P^{m}(X\leq a|C_{m,i}),
\end{align}
which completes the proof of~\eqref{part_1}.

\subsection{Proof of (\ref{part_2})}\label{a:part2}

By the definition of limit,  proving~\eqref{part_2} equals to prove $\forall \lambda \ge 0, a \in \{0,1,\dots\}, \exists N_{0}$ such that $\forall N >N_{0}$,
$|P^{N}(X \le a) - P(X \le a)| < \lambda$.

Conditioning on the occurrence of event $C_N$,
\begin{align}
    P^{N}(X \le a)
    &= P^{N}(C_{N}^{c})P^{N}(X\le a | C_{N}^{c}) + P^{N}(C_{N}) P^{N}(X\le a | C_{N}),\\
    P(X \le a)
    &= P(C_{N}^{c})P(X  \le a | C_{N}^{c})+ P(C_{N}) P(X \le a | C_{N}).
\end{align}
By definition of~\eqref{eq:Pm1} and~\eqref{eq:Pm2}, $P(X \le a  | C_{N}^{c} ) =  P^{N}(X  \le a | C_{N}^{c})$ and $P(C_{N}^{c}) = P^{N}(C_{N}^{c})$.  Let $p_1 =  P^{N}(X \le a | C_{N})$ and $p_2 = P(X \le a  | C_{N} )$.  We have,
\begin{align}
    | P^{N}(X \le a) - P(X \le a) | 
    &= P(C_{N}) \, | p_1 - p_2 |\\
    &\le (1-\delta)\left(  1-\delta(1-\epsilon)
       \right)^{N-1}| p_1 - p_2 |\label{eq:1-P(C^{c})}\\
    &< \lambda,
\end{align}
where~\eqref{eq:1-P(C^{c})} is due to the following
\begin{align}
    P(C_{N})\leq P^{N-1}(C_{N})\leq P^{N-2}(C_{N})\leq \cdots &\leq P^{0}(C_{N}), \label{eq:inductive}
\end{align}
\eqref{eq:inductive} can be derived by definition of~\eqref{eq:Pm1} and~\eqref{eq:Pm2} as 
\begin{align}
    P^{i}(C_N)
    &= P^{0}(C_N|C_{i+1})P^{i}(C_{i+1}|C_i)P^{i}(C_i) \\
    &\leq
 P^{0}(C_N|C_{i+1})P^{0}(C_{i+1}|C_i)P^{i-1}(C_i).
\end{align}
Then, by the memoryless property of the Markov process with measure $P^{0}(\cdot)$we have
\begin{align}
      P^{0}(C_{N})
      & = P^{0}(C_1|U_0)\prod_{i = 1}^{N-1}P^{0}(C_{i+1}|C_{i})\\
      &= (1-\delta)\left(  1-\delta(1-\epsilon)
       \right)^{N-1},
\end{align}
where for $i=\{1,2,\cdots,N-1\}$
\begin{align}
    P^{0}(C_{i+1}|C_{i}) &= 1-P^{0}(A_{i}|C_i)\\
    &=1-\delta(1-\epsilon) .
\end{align}
Thus  $\lim\limits_{N \rightarrow \infty} P^{N}(X \le a) = P(X \le a)$.  Hence~\eqref{part_2} holds.

\section{Proof of Lemma \ref{lm:x+y}}
\label{a:x+y}

For every integer $n$, we have
\begin{align}
    P(X+Y \geq 0) &\leq P(X > -n \text{ or } Y > n -1)\\
    & \leq  P(X > -n) + P(Y > n -1) .
\end{align}
Taking the infimum over $n$ on both sides yields the upper bound in Lemma~\ref{lm:x+y}.

We prove the other direction in the following.
Denote $\phi (n) = \bar{F}_{X}(-n) + \bar{F}_{Y}(n-1) $, and we have $0 \leq \phi \leq 2$. Since $\phi(n) \to 1$ as $|n| \to \infty$, $\inf_n \phi(n) \le 1$. If $\inf_n \phi(n) = 1$, the bound is trivial. Hence in the rest of the proof we may assume that $\phi(a) = \inf_n \phi(n) < 1$ for some $a \in \mathbb{Z}$. 

Then for all $k \in \mathbb{Z}$, we have
\begin{align} \label{leq}
    \bar{F}_{X}(-a+k) + \bar{F}_{Y}( a-k-1) \geq \bar{F}_{X}( -a)+\bar{F}_{Y}( a-1)
\end{align}
Or equivalently, for $k \in \mathbb{N}$,
\begin{align}\label{nya}
    P(-a < X \leq -a+k) \leq P(a-k-1 < Y \leq a-1),
\end{align}
\begin{align}\label{nyanya}
    P(a-1 < Y \leq a+k-1) \leq P(-a-k < X \leq -a).
\end{align}

To attain the bound, we must construct $X_0$ and $Y_0$ such that $X_0+Y_0 \ge 0$ whenever $X_0 > -a$ or $Y_0 > a-1$, and we need to ensure the two cases are disjoint. 

Set a random variable $U$, which has a uniform distribution over $[0,1]$. Given $F$, define 
\begin{align}
    F^{-1}(x) = \sup \{a \in \mathbb{R}|F(a) < x\}.
\end{align}

We define $X_0$ and $Y_0$ to be the following:

\begin{align}
    X_0 = F_X^{-1}(U+F_X(-a)-\lfloor U+F_X(-a) \rfloor)
\end{align}
\begin{align}
    Y_0 = F_Y^{-1}(F_Y(a-1)-U-\lfloor F_Y(a-1)-U \rfloor)
\end{align}

   $X_0$ and $Y_0$ have cdf $F_X$ and $F_Y$ since they are rearrangements of $F^{-1}_X$ and $F^{-1}_Y$. Now, we show that $X_0(u) + Y_0(u) \geq 0$ for $u \leq 1 - F_X(-a)$ $a.e.$ The same holds for $u > F_Y(a-1)$, and we omit this similar proof here.
\begin{align}
    F^{-1}_X(x) = \sup \{a \in \mathbb{R}|F_X(a) < x\}.
\end{align}
implies that 
\begin{align}
F_{x} \circ F_{x}^{-1}(x)>x \quad a.e. ,
\end{align}
since $X$ is integer-valued.
Consider 
\begin{align}
S_{b}=\left\{u \leq 1-F_{X}(-a) | X_0(u) \leq b, Y_0(u) \leq -b-1\right\}
\end{align}
And,
\begin{align}
    F_X(b)
    &\geq F_X \circ X_0(u) \\
    &= F_X \circ F^{-1}_X(u+F_{X}(-a)) \\
    &> u+F_{X}(-a) \quad a.e. \label{fb},
\end{align}
Accordingly,
\begin{align}
    F_Y(-b-1) 
    &\geq F_Y \circ Y_0(u) \\
    &= F_Y \circ F^{-1}_Y(F_{Y}(a-1)-u) \\
    &> F_{Y}(a-1)-u \quad a.e. \label{fb1}
\end{align}
Adding up~\eqref{fb} and~\eqref{fb1} yields
\begin{align}
    F_X(b)+F_Y(-b-1) > F_{X}(-a)+F_{Y}(a-1) \quad a.e. 
\end{align}
However, $\phi (a)$ is maximal, leading to $P(S_b) = 0$
\begin{align}
    P(u \leq 1- F_X(-a) | X_0(u)+Y_0(u)<0) = 0
\end{align}
Finally, 
\begin{align}
    P(X_0+Y_0) & \geq 1-F_X(-a)+1-F_Y(a-1) \notag\\
    &\quad= \bar{F}_{X}(-a)+\bar{F}_{Y}(a-1) 
\end{align}

Thus $X_0$ and $Y_0$ attain the optimal bound, which concludes the lemma.

\section{Proof of Theorem~\ref{th:Cher_upper}}
\label{s:Cher_upper}

This appendix finished the proof of Theorem~\ref{th:Cher_upper} in Sec.~\ref{subsec:sketch}. We 
utilize Theorem~\ref{th:X}, proved in Appendix~\ref{s:X}, to derive 
a closed-form upper bound on the probability of a safety violation for confirmation by depth 
$k$.  We produce two upper bounds, both of which decay exponentially with $k$.  The first bound 
has a weaker exponent, but it is easy to verify its decay within the fault tolerance~\eqref{eq:a<}.  The second bound 
has a steeper exponent whose expression is more complicated.

\subsection{Upper Bounding $Y$}
\label{a:s:Y}
Recall that $Y^{(k)}$ is defined in \eqref{eq:Yk}.
Due to the renewal nature of the pacer process, we have $P_{\tau+\Delta,\tau+\Delta+c}$ stochastically dominating $P'_{0,c}$, where $P'$ is an independent pacer process of the same statistics as $P$.  Hence
\begin{align}
    &P(  Y^{(k)} > -n ) \notag \\
    &\le P\!\big( L_s - k + \sup_{c \geq 0} \{A_{s,\tau+2\Delta+c}-P'_{0,c} \} > -n \big) \\
    &\le P\!\big( L_s - k + A_{s,\tau+2\Delta} + \sup_{c \geq 0} \{A_{\tau+2\Delta,\tau+2\Delta+c}-P'_{0,c} \}  > -n \big) .  \label{eq:LAsup}
\end{align}
The lead $L_s$ depends only on blocks mined by $s$.  The supreme term in~\eqref{eq:LAsup} depends only on the mining processes after $\tau+2\Delta>s$.  Hence the supreme term in~\eqref{eq:LAsup} is independent
of $L_s - k + A_{s,\tau+2\Delta}$.

Let us define
\begin{align}\label{eq:define M}
    M= \sup_{t\geq 0}\{A'_{0,t}-P'_{0,t}\}    
\end{align}
where $A'$ is a Poisson process identically distributed as $A$.
By Lemma~\ref{lm:LS_trans}, the pmf of $M$ is given by~\eqref{eq:e(i)}.  Evidently, the supreme term in~\eqref{eq:LAsup} is identically distributed as $M$.  We can thus rewrite~\eqref{eq:LAsup} as
\begin{align}
    P( Y^{(k)} > -n )
    &\le P\left( L_s - k + A_{s,\tau+2\Delta} + M >-n\right) .  \label{eq:LAsupM}
\end{align}

Let the first H-block mined after 
 $s$ be called the first new pacer (this takes an exponential time), then the first H-block mined at least $\Delta$ seconds after that be called the second new pacer, and so on.  Let us also define 
 $s+T-\Delta$ to be the first time $k$ new pacers are mined.  Since the height of the $k$-th new pacer is at least $h_b+k$, we have  $s+T-\Delta\ge\tau$ by $\tau$'s definition.  We then relax~\eqref{eq:LAsupM} to

\begin{align}
    P( Y^{(k)} > -n )
    &\le P\left( L_s - k + A_{s,s+T+\Delta} + M > -n \right) .  \label{eq:LAsupT}
\end{align}

While $\tau$ is generally dependent on $L_s$, we have $T$ independent of $L_s$ due to the memoryless nature of mining.

In Appendix~\ref{a:lead}, we show that $L_s$ is stochastically dominated by $M+A_{0,2\Delta}$.  Let $L$ be an independent random variable identically distributed as $M$.  We also note that
$ A_{s,s+T+\Delta}$ is identically distributed as $A_{2\Delta,T+3\Delta}$.  We can simplify~\eqref{eq:LAsupT} to

\begin{align}
    P( Y^{(k)} > -n )
    &\le P\left( L + A_{0,2\Delta} - k + A_{2\Delta,T+3\Delta} + M > -n \right) \\
    &= P\left( L + M + A_{0,T+3\Delta} - k + n > 0 \right) \label{eq:M+L+A4d}.
\end{align}


In~\eqref{eq:M+L+A4d}, $L$, $M$, and $A_{0,T+3\Delta}$ are mutually independent.  We shall derive their MGFs shortly and use them to upper bound the probability on the right hand side of~\eqref{eq:M+L+A4d}.  We take two separate approaches: In Appendix~\ref{s:cor:tolerance}, we let $n=ck$ with a specific choice of $c>0$ and establish Corollary~\ref{cor:tolerance} by showing that as long as~\eqref{eq:a<} holds,~\eqref{eq:PXn} and~\eqref{eq:M+L+A4d} both vanish as $k\to\infty$.


\subsection{Chernoff Bound}
\label{s:proof_cher}

In this appendix, 
obtain an explicit exponential bound and thus establish Theorem~\ref{th:Cher_upper}.  
(This analysis can be thought of as a refinement of the analysis in Appendix~\ref{s:cor:tolerance}.)
One step is to optimize $\nu$ introduced for the Chernoff bound.  A second step is to select $n$ carefully to balance the two exponents corresponding to $P(X^{(k)}\ge n)$ and $P(Y^{(k)}>-n)$, respectively.

First, by Theorem~\ref{th:X},
\begin{align}
    P(X^{(k)}\ge n) &\leq P(X\ge n)\\
& \le (1-\delta )\left(\frac{\epsilon}{\epsilon+\delta-\epsilon\delta}\right)^n
    \label{eq:PXk>n-1} 
\end{align}
where $\epsilon$ and $\delta$ are defined in~\eqref{eq:epsilon} and~\eqref{eq:delta}, respectively. Note that this even holds for non-integer $n$.

We can use Chernoff inequality again
to upper bound $P(Y^{(k)}>-n)$.
Recall that $\bar{a}$ and $\bar{h}$ are defined as $a\Delta$ and $h\Delta$, respectively.
Using the same computation in Appendix~\ref{s:cor:tolerance}, we have
\begin{align}
    P( Y^{(k)}>-n)
    &\le P(M+L+A_{0,T+3\Delta}+n-k >0) \\
    &\leq e^{(e^{\nu}-1)3\bar{a}}
    \left( q(\nu) \right)^k
    e^{\nu n}E[e^{\nu M}]E[e^{\nu L}]\label{eq:Chernoff_upper}
\end{align}
where
\begin{align}
    q(\nu) = \frac{e^{(e^{\nu}-1)\bar{a}-\nu}}{1-(e^{\nu}-1)\bar{a}/\bar{h}} .
\end{align}
We select $\nu$ to minimize $q(\nu)$, which is not difficult by setting the derivative $q'(\nu)=0$.
The unique valid solution $\nu^{*}$ satisfies
\begin{align} \label{eq:ev*}
    e^{\nu^{*}} = 1+\gamma(\bar{a},\bar{h})/\bar{a}
\end{align}
where 
\begin{align}
    \gamma (\alpha,\eta) = \frac12\left(2-\alpha+\eta - \sqrt{4+(\alpha+\eta)^2} \right).
\end{align}

Letting $n=\zeta k$ with $\zeta$ to be decided later, we have 
\begin{align}
    P(X^{(k)}\ge \zeta k) \leq   
     (1-\delta )e^{\zeta k\, \ln\left(\frac{\epsilon}{\epsilon+\delta-\epsilon\delta}\right)} .
\end{align}

Using lemma~\ref{lm:LS_trans}, we evaluate the MGFs of $M$ and $L$ as
\begin{align}
    E[e^{\nu^{*} L}]
    &= E[e^{\nu^{*} M}] \\
    &= \mathcal{E}\left( 1+ \frac{\gamma(\bar{a},\bar{h})}{\bar{a}} \right) .
\end{align}

Letting $n=\zeta k$ and plugging~\eqref{eq:ev*} back into~\eqref{eq:Chernoff_upper} 
and also evaluating the MGFs of $M$ and $L$ , we obtain 
\begin{align}\label{eq:PYk>-n < fe...}
    P( Y^{(k)} > -n )
    \leq f(\bar{a},\bar{h}) e^{(\zeta \nu^*-\xi(\bar{a},\bar{h})) k}
\end{align}
where
\begin{align}
    f(\alpha,\eta) = \frac{e^{5\gamma}\gamma^2(\alpha-\eta+\alpha \eta)^2}{((\alpha+\gamma)(\gamma-\eta)+\alpha \eta e^{\gamma})^2}
\end{align}
and
\begin{align}
    \xi(\alpha,\eta) 
    &= -\gamma - 
    \ln \frac{\eta \alpha }{(\alpha+\gamma)(\eta-\gamma)}\\
    &=   
    \ln \left( 1 + \frac{\gamma}{\alpha} \right)
    + \ln \left(1-\frac{\gamma}{\eta} \right)
    - \gamma. 
\end{align}

To strike a balance between the exponents of $P(Y^{(k)}>-\zeta k)$ and $P(X^{(k)}\ge \zeta k)$, we let
\begin{align}
    \zeta \nu^{*} -  \xi(\bar{a},\bar{h}) = \zeta \ln(\frac{\epsilon}{\epsilon+\delta-\epsilon\delta}).
\end{align}
It gives us
\begin{align}
    \zeta = \frac{\xi(\bar{a},\bar{h})}{\ln(1+\frac{\gamma}{\bar{a}})  -\ln (\frac{\epsilon}{\epsilon+\delta-\epsilon\delta})} ,
\end{align}
such that 
\begin{align}
      P( Y^{(k)} > -\zeta k )
      \leq  f(\bar{a},\bar{h})(e^{-c(\bar{a},\bar{h})})^k,
\end{align}
and
\begin{align}
    P( X^{(k)} \ge \zeta k )
    \leq (1-\delta) (e^{- c(\bar{a},\bar{h})})^k,
\end{align}
with
\begin{align}
    c(\bar{a},\bar{h}) &= \xi(\bar{a},\bar{h})-\nu^{*}\zeta\\
     &= \xi(\bar{a},\bar{h}) (1  - \ln (1+\frac{\gamma}{\bar{a}}) (\ln (\frac{\epsilon}{\epsilon+\delta-\epsilon\delta}))^{-1}   )^{-1},
\end{align}
which can be seen from the exponential parameter of~\eqref{eq:PYk>-n < fe...}.

Denote 
\begin{align}
   b(\bar{a}, \bar{h}) &= f(\bar{a}, \bar{h})+1-\delta \\
   & = \frac{e^{5\gamma}\gamma^2(\bar{a}-\bar{h}+\bar{a} \bar{h})^2}{((\bar{a}+\gamma)(\gamma-\bar{h})+\bar{a} \bar{h} e^{\gamma})^2} +1 - \left(1-\bar{a}\frac{1+\bar{h}}{\bar{h}}\right)e^{-\bar{a}}.
\end{align}
We have 
\begin{align}
    P(X^{(k)} \geq \zeta k)+ P(Y^{(k)}> -\zeta k) 
    \leq b(\bar{a},\bar{h}) e^{-c(\bar{a},\bar{h})k},
\end{align}
which completes the proof of Theorem~\ref{th:Cher_upper}.

In Appendices~\ref{a:upper_depth} and~\ref{a:lower_depth}, we evaluate the error events in better accuracy than the Chernoff bounds allow us to develop a pair of upper and lower bounds for the probability of safety violation.  The results are quite easy to evaluate numerically.

\section{Bounding the Lead}
\label{a:lead}



The lead increases by at most one upon the mining of an A-block and decreases by at least one upon the publication of a jumper unless the lead is already zero, in which case it stays at 0 with the publication of a jumper.  Note that $L_0=0$.  On every sample path, there exists $u\in[0,s]$ with $L_u=0$ such that
\begin{align}
    L_s - L_u  
    &\le A_{u,s} - J_{u,s-\Delta} . \label{eq:L<=A-J}
\end{align}

If $L_s=0$, $u=s$ satisfies~\eqref{eq:L<=A-J}.  If $L_s>0$, $u$ is in the last interval of zero lead before $s$.
In general, $L_u=0$ and~\eqref{eq:L<=A-J} imply that
\begin{align}
    L_s
    &\le \sup_{q\in[0,s]} \{ A_{q,s} - J_{q,s-\Delta} \} \\
    &\le \sup_{q\in [0,s]} \{ A_{q,s} - P_{q+\Delta,s-\Delta} \} \label{eq:L<=A-P} \\
    &\le \sup_{r\in [0,s-\Delta)} \{ A_{s-\Delta-r,s-2\Delta} - P_{s-r,s-\Delta} \} + A_{s-2\Delta,s} . \label{eq:L<=APA}
\end{align}

where we have invoked Lemma~\ref{jp} to obtain~\eqref{eq:L<=A-P} and also used a change of variable $r=s-\Delta-q$ to obtain~\eqref{eq:L<=APA}.

Let us introduce a reverse pacer process, $P^{s-\Delta}_{0,r}$, which denotes the number of pacers mined during $(s-r,s-\Delta]$, i.e.,
\begin{align}
    P^{s-\Delta}_{0,r} = P_{s-r,s-\Delta} .
\end{align}

Due to the nature of the pacer process, one can show that 
$P^{s-\Delta}_{0,r}$ stochastically dominates $P_{0,r-\Delta}$
as a process.  Since $A$ is a memoryless Poisson process, $A_{s-\Delta-r,s-2\Delta}$ is statistically the same process as $A_{0,r-\Delta}$.
Recall also $A'$ is statistically identical to $A$.  
We can use~\eqref{eq:L<=APA} to show that $L_s$ is stochastically dominated in the sense that for every real number $x$,
\begin{align}
    &P( L_s \ge x ) \notag\\
    &\le  P\bigg( \sup_{r\in [0,s-\Delta)} \{ A_{s-\Delta-r,s-2\Delta} - P_{0,r-\Delta} \} + A_{s-2\Delta,s} \geq x \bigg) \\
   & \leq P\bigg( \sup_{r\in [0,s-\Delta)} \{A_{0,r-\Delta} - P_{0,r-\Delta} \} +A'_{t,t+2\Delta} \geq x \bigg) \\
   & \leq   P\bigg( \sup_{q\ge0}
    \{ A_{0,q} - P_{0,q} \} + A'_{t,t+2\Delta} \geq x \bigg)\\
   &=  P(L + A'_{t,t+2\Delta} \geq x) .
\end{align}

\section{Proof of Corollary~\ref{cor:tolerance}}
\label{s:cor:tolerance}

This proof follows Appendix~\ref{a:s:Y}. It suffices to show that for every $\epsilon>0$,
there exists $c>0$ such that
\begin{align}
    P(X^{(k)} \ge ck) &\leq \epsilon/2  
    \label{eq:PXk>} \\
    P(Y^{(k)} >-ck) &\leq \epsilon/2 
    \label{eq:PYk>}
\end{align}
for sufficiently large $k$.  
By Theorem~\ref{thm:balanced heights},~\eqref{eq:PXk>} holds for large enough $k$ for every $c>0$.  The remainder of this appendix is devoted to the proof of~\eqref{eq:PYk>}.

By~\eqref{eq:M+L+A4d} and the Chernoff bound, we have for 
$\nu>0$,
\begin{align}
    P( Y^{(k)} >-ck)
    &\leq P(L+M+A_{0,T+3\Delta} >(1-c)k) \label{eq:L+M+A+(1-c)k}\\
    &\leq E [e^{\nu(M+L+A_{0,T+3\Delta}+(c-1)k)}] \\
    &= E[ e^{\nu(A_{0,T+3\Delta}+(c-1)k)}]E[e^{\nu M}]E[e^{\nu L}]
    \label{eq:EAEMEL}
\end{align}
where $E[e^{\nu L}] = E[e^{\nu M}] = \mathcal{E}(e^\nu)$ is finite for every $\nu>0$
by Lemma~\ref{lm:LS_trans}, where $\mathcal{E}$ is defined as~\eqref{eq:E(r)}.

Only the first term on the right-hand side of
~\eqref{eq:EAEMEL} depends on $k$.  Given $t\ge0$, $A_{0,t}$ is a Poisson random variable with mean $at$, whose MGF is
\begin{align} \label{eq:mgf:poisson}
    E[e^{\nu A_{0,t}}] = e^{\theta t}
\end{align}
where we have defined
\begin{align}
    \theta = \left(e^\nu - 1\right) a
\end{align}
for convenience.
By definition, $T$ is equal to $k\Delta$ plus the sum of $k$ i.i.d.\ exponential($h$) variables.  The MGF of an exponential random variable $Z$ with parameter $h$ is
\begin{align} \label{eq:mgf:exp}
    E[e^{\eta Z}] = \frac1{1-\eta/h} .
\end{align}
Using~\eqref{eq:mgf:poisson},~\eqref{eq:mgf:exp}, and the definition of $T$, we can write
\begin{align}
    E[ & e^{\nu(A_{0,T+3\Delta}+(c-1)k)}] \notag \\ 
    &= E[ E[e^{\nu(A_{0,T+3\Delta})}|T] ] \cdot e^{\nu(c-1)k} \\
    &= E[ e^{\theta(T+3\Delta)} ] \cdot e^{\nu(c-1)k} \\
    &= \left( E[ e^{\theta Z} ] \right)^k \cdot e^{\theta(3+k)\Delta} \cdot e^{\nu(c-1)k}\\
    &= e^{3\theta\Delta} ( g(\nu,c) )^{-k} \label{eq:g^k}
\end{align}
where
\begin{align} \label{eq:under k}
    g(\nu,c)
    &= \left( 1-\frac{\theta}h \right)
    \left( 1+\frac{\theta}a \right)^{1-c}
    e^{-\theta\Delta} . 
\end{align}

We let
\begin{align} \label{eq:c=/4}
    c =
    \frac{a}2 \left( \frac1a - \frac1h - \Delta \right) ,
\end{align}
which is strictly positive when~\eqref{eq:repeat a<} holds, i.e., when the parameters are within the fault tolerance.
Using Taylor series expansion, we have
\begin{align}
    g(\nu,c)
    &=
    1 + \left( \frac{1-c}{a} - \frac1{h} - \Delta \right) \theta + o(\theta) \\
    &=
    1 + \frac12 \left( \frac1a - \frac1h - \Delta\right) \theta + o(\theta) . \label{eq:gnuc1}
\end{align}
The coefficient of the first-order term in~\eqref{eq:gnuc1} is strictly positive due to~\eqref{eq:repeat a<}. 
Hence there exists $\theta>0$ for which $g(\nu,c)>1$,
so that the right hand side of~\eqref{eq:g^k} and thus also the right hand side of~\eqref{eq:EAEMEL} vanish as $k\to\infty$. 
For sufficiently large $k$,~\eqref{eq:PYk>} holds.  
The proof of Corollary~\ref{cor:tolerance} is thus complete.

\section{Proof of Theorem~\ref{th:upper_depth}}\label{a:upper_depth}

To derive a tighter upper bound in the case of confirmation by depth, we evaluate $P(Y^{(k)}>-n)$ without resorting to the Chernoff bound.  Starting from~\eqref{eq:M+L+A4d} in Appendix~\ref{s:Cher_upper}, we use the law of total probability to write
\begin{align}
    P( & A_{0,T+3\Delta} \leq -n+k-(M+L)) \notag \\
    =& \sum_{i=0}^{k-n}P(L=i)(\sum_{j=0}^{k-n-i}P(M=j) P(A_{0,T+3\Delta} \leq -n+k-(i+j)))\\
    =& \sum_{i=0}^{k-n}P(L=i)\sum_{j=0}^{k-n-i}P(M=j)\int_{0}^{\infty}P(A_{0,T+3\Delta} \leq -n+k-(i+j)|T = t)f_{T}( t)\, dt\label{eq:PA-int}\\
    =& \sum_{i=0}^{k-n}P(L=i)\sum_{j=0}^{k-n+1-i}P(M=j)\int_{0}^{\infty}P(A_{0,t+3\Delta} \leq -n+k-(i+j))f_{T}( t)\, dt .
\end{align}
Note the pdf of the $k$-st new pacer's mining time $T$ , $f_{T}(t)$ can be driven by taking the derivative of its cdf , which is
\begin{align}
    F_{2}(t-k\Delta;k,h),
\end{align}
where $F_{2}$ denotes the cdf of Erlang distribution.
Now we only have to calculate the last term $P(A_{0,t+3\Delta} \leq -n+k-(i+j))$, which equals
\begin{align}
 F_1(-n+k-(i+j);a(t+3\Delta))).
\end{align}

Finally, we have 
\begin{align}
 &   P(Y^{(k)}\leq -n)\notag \\
    \geq & \sum_{i=0}^{k-n}P(L=i)\sum_{j=0}^{k-n-i}P(M=j)\int_{0}^{\infty}F_1(-n+k-(i+j);a(t+3\Delta))f_{2}(t-k\Delta;k,h)\, dt\label{eq:PY<n}\\
  =& \sum_{i=0}^{k-n}e(i)\sum_{j=0}^{k-n-i}e(j) \int_{0}^{\infty}F_1(-n+k-(i+j);a(t+3\Delta)) \cdot f_{2}(t-k\Delta;k,h)\, dt.\label{eq:PY<n_e}\
\end{align}
Thus, according to~\eqref{eq:inf(FX+FY)}, the probability of violation is upper bounded by
\begin{align}
    \inf_{n}\{2-P(Y^{(k)} \leq -n)-P(X^{(k)}\leq n-1)\},
  \label{eq:P(F_sk)}
\end{align}
which equals~\eqref{eq:upper_depth}. This completes the proof of Theorem~\ref{th:upper_depth}.

\section{Proof of Theorem \ref{th:lower_depth}}\label{a:lower_depth}
Consider the same attack as the private-mining attack in Definition~\ref{def:private-mining} except that we let the propagation delays of all H-blocks be zero before time $s$ for simplicity.  The adversary mines privately from  $s$ and delays all new H-blocks maximally.
Let the first H-block mined after $s$ be called the first new pacer (this takes an exponential time), then the first H-block mined at least $\Delta$ seconds after that be called the second new pacer, and so on.  Let $P'_r$ denote the number of new pacers mined from $s$ up until time $s+r$.  Evidently, $P'_r$ is identically distributed as $P_{0,r+\Delta}$. 

A violation occurs if 
there exists $d\geq s+\tau$ such that 1) the lead $L_s$ plus the number of A-blocks 
$A_{s,d}$ is no fewer than $P'_{d-s-\Delta}$, the number of new pacers mined between $s$ and $d-\Delta$ (any additional ones mined after that is not yet public);
and 2) at least one pacer arrives by time $d$ to be the target block, i.e., 
 $P'_{d-s-\Delta} \geq 1$. 
Hence the probability of violation is lower bounded by
\begin{align}
    P & \left( 
    \cup_{d \geq s+\tau}\left\{L_{s} + A_{s,d}  \geq 
P'_{d-s-\Delta} \ge 1
\right\} 
 \right)\label{eq:L+A>P>1} ,
 \end{align}
 where $\tau$ is the time it takes for $k$ such new pacers to be mined after $s$. 
 For $k \geq 1$, $P'_{d-s-\Delta} \geq 1$.
Since $L_{s}$, $A_{s,d}$ and  $P'_{d-s-\Delta}$ are independent, we have the following similar in Appendix~\ref{s:Cher_upper},
\begin{align}
  &  P(\bigcup_{d \geq s+\tau}\left\{L_{s} + A_{s,d} \geq 
P'_{d-s-\Delta}
\right\})\notag \\
& = P(L_{s}+\max_{d \geq s+\tau}\{A_{s,d} - P'_{d-s-\Delta}\} \geq 0)\\
& = P(L_{s}+\max_{d \geq \tau}\{A'_{0,d} - P'_{d-\Delta}\} \geq 0)\\
& \geq  P(L_{s}+A'_{0,\tau} + \max_{d \geq \tau}\{A'_{\tau,d} - P'_{0,d-\tau}\} -k\geq 0)\label{eq:change P'to k}\\
&= P(0\leq A'_{0,\tau}+M+L_{s}-k)\label{eq:P'<A'+M+Ldepth},
\end{align}
where~\eqref{eq:change P'to k} is due to that there are $k$ new pacers during $(s,\tau]$. $M$ in~\eqref{eq:P'<A'+M+Ldepth} is defined in~\eqref{eq:define M}.
$P'$ and $P$ are i.i.d. 
 Thus,~\eqref{eq:P'<A'+M+Ldepth} can be written as
\begin{align}
\begin{split}
   P&(\bigcup_{d \geq s+\tau}
   \left\{L_{s} + A_{s,d} \geq P'_{d-s+\Delta}\right\}) \\
   & \geq 1-P(A'_{0,\tau} \leq k-1-M-L_{s})     .
\end{split}
\end{align}
Now, we want to calculate $P(A'_{0,\tau}\leq k-1-L_{s}-M)$ .
We have
\begin{align}
   & P(A'_{0,\tau}\leq k-1-L_{s}-M) \notag \\
   =& \sum_{i=0}^{k}\sum_{j=0}^{k-i}P(A'_{0,\tau}\leq k-1-L_{s}-M)P(M=j)P(L_{s} = i),
\end{align}
where $P(A'_{0,\tau}\leq k-1-L_{s}-M)$ can be calculated as 
\begin{align}
    \sum_{i=0}^{k} \Big(1-\frac{a}{h}\Big) \Big(\frac{a}{h}\Big)^i
    \sum_{i=0}^{k-i}e(j)\int_{0}^{\infty}P(A_{0,t}\leq k-1-i-j)f_{\tau}(t)\, dt,
\end{align}
where $\tau$ has cdf
$
    F_2(t-k\Delta;k,h)
$, 
and
\begin{align}
    P(A_{0,t}\leq k-1-i-j) = F_1(k-1-i-j;at).
\end{align}
Finally, the probability that the target transaction's safety is violated is no large than
\eqref{eq:lower_depth}. To ease computation, we can also relax the infinite integral to finite integral as 
\begin{align}
\begin{split}
   & \int_{0}^{\infty} F_1(k-1-(i+j);at) \cdot f_{2}(t-k\Delta;k,h)\, dt \\
   & \leq 
    \int_{0}^{T} F_1(k-1-(i+j);at) \cdot f_{2}(t-k\Delta;k,h)\, dt \\
    &\quad +\Bar{F}_2(T;k,h)\label{eq:relax_depth_lower} .    
\end{split}
\end{align}

\section{Proof of Theorem~\ref{th:upper_time}}
\label{a:upper_time}


There are ample similarities between confirmation by time and confirmation by depth.  For example, the distribution of the adversarial lead $L_s$ remains the same.  Also, we shall express the error event in terms of the race between the pacers and the A-blocks, where the latter is helped by the lead and the balanced heights.  However, the starting time of the error events of interest is different than in the case of confirmation by depth.  

\subsection{An Enlarged Violation Event} 
    
\begin{Def}[Unbalanced jumper]
  We let $\hat{J}_{s,t}$ denote the number of jumpers mined during $(s,t-\Delta]$ which are unbalanced.  
\end{Def}
The sequence of jumpers includes all candidates, some of which may be balanced. Whether a jumper mined at time $s$ is balanced or unbalanced is known by time $s+\Delta$.

\begin{lemma}
\label{disag}
    Suppose $r' \geq r \geq s\geq 0$. If an $r$-credible chain and an $r'$-credible chain disagree on a target block that is 
 mined by time $s$, then 
    \begin{align}\label{eq:disagree_rr'}
    L_{s} + A_{s,r'}  \geq \hat{J}_{s,r}
    \end{align}
\end{lemma}


\begin{proof} 
    If $s \geq r-\Delta$, we have $\hat{J}_{s,r} =0$, so~\eqref{eq:disagree_rr'} holds trivially.  We assume $s < r-\Delta$ in the remainder of this proof.
    Since the target block is mined by time $s$,
    there must be disagreements on all heights of jumpers mined from $s$ to $r-\Delta$, i.e.,
    on no fewer than $J_{s,r-\Delta}$ heights. Let us denote the public height at time $s$ as height $\eta$. Every unbalanced jumper mined during $(s,r-\Delta]$ must meet either one of the following two conditions:

(1) All honest blocks on their height are in agreement, so all blocks on that height in a disagreeing branch must be adversarial;

(2) Some honest blocks on their height disagree, but at least one A-block is mined on that height by time $r'$.

In either case, at least one A-block must be mined by $r'$ on the height of every unbalanced jumper mined during $(s,r-\Delta]$. And the number of those unbalanced jumpers mined during $(s,r-\Delta]$ is $\hat{J}_{s,r}$.

The adversary needs to mine at least on $\hat{J}_{s,r}$ heights of  unbalanced jumpers. Due to the pre-mining phase, at time $s$, the highest A-block has a height $\eta+L_{s}$, which is at most $L_{s}$ higher than the highest honest block, which means at most $L_{s}$ of A-blocks can be mined before time $s$ on the heights of unbalanced jumpers. The rest A-blocks can only be mined after time $s$, the number of which is  $A_{s,r'}$.  Therefore, the disagreement about the target transaction implies the inequality~\eqref{eq:disagree_rr'}.  
Hence the lemma is proved.
\end{proof}

\begin{lemma} 
Fix $\epsilon \in (0,\Delta)$, barring the following event 
\begin{align}
F_{s,t}^{\epsilon}=\bigcup_{d \geq t+s}\left\{L_{s} + A_{s,d}  \geq J_{s,d-\epsilon-\Delta} - X
\right\}\label{error}
\end{align}
all $r'$-credible chains agree on the target block for all $r' \geq t+s$.
\end{lemma}
\begin{proof}
Let $r = t+s$.
We first establish the result for $r' \in [r , r + \epsilon ]$ by contradiction. Then prove the lemma by induction.
Fix arbitrary $r' \in [r , r + \epsilon ]$. If an $r'$-credible chain and an $r$-credible chain disagree, then we have the following by Lemma \ref{disag}:
\begin{align}
     L_{s} + A_{s,r'}  &\geq \hat{J}_{s,r} \\
     &\geq \hat{J}_{s,r'-\epsilon}\\
    &\geq J_{s,r'-\epsilon-\Delta}-X\label{contrad}
\end{align}
where~\eqref{contrad} is due to that the number of balanced candidates mined during $(s,r'-\epsilon]$ is smaller than $X$.
If~\eqref{contrad} occurs, 
$F_{s,t}^{\epsilon}$ also occurs. So this implies that if $F_{s,t}^{\epsilon}$ does not occur, every $r'$-credible chain agrees with the $r$-credible chain, proving the lemma in this case.

Now suppose this lemma holds for $r' \in [r , r + n\epsilon ]$ for some positive integer $n$. We show the lemma also holds for $r' \in [r , r + (n+1)\epsilon]$ as the following: First, let $r'' = r+n\epsilon$.
A repetition of the contradiction case above with $r$ replaced by $r''$ implies that the lemma also holds for $r' \in [r'',r''+\epsilon]$. Therefore lemma holds for $r' \in [r , r + (n+1)\epsilon]$. The lemma is then established by induction on $n$.
\end{proof}

By Lemma \ref{jp}, we have
\begin{align}
        J_{s,d-\epsilon-\Delta} -  X
        &\geq P_{s+\Delta,d-\epsilon-\Delta} - X.
        \label{eq:JXPX}
\end{align}
By~\eqref{error} and~\eqref{eq:JXPX}, we have
\begin{align}
    F_{s,t}^{\epsilon} \subset \bigcup_{d \geq s+t}\left\{L_{s} + A_{s,d}  \geq P_{s+\Delta,d-\epsilon-\Delta} -  X\right\} .
\end{align}
Hence
\begin{align}
    P(F_{s,t}^{\epsilon}) &\leq P(\bigcup_{d \geq s+t}\left\{L_{s} + A_{s,d}  \geq P_{s+\Delta,d-\epsilon-\Delta} -  X\right\}) \\
    &= P(\sup_{d \geq s+t}\{L_{s} + A_{s,d} -P_{s+\Delta,d-\epsilon-\Delta} \}+X \geq 0),
\end{align}
where equality is by the definition of the union.

Let us define
\begin{align}
    Y_t = \sup_{d \geq s+t}\{L_{s} + A_{s,d} -P_{s+\Delta,d-\epsilon-\Delta} \} .
\end{align}
We apply Lemma~\ref{lm:x+y} to obtain
    \begin{align}
    P(F_{s,t}^{\epsilon}) \leq \inf_{n \in \{1,2,... \}}\{ P(X \ge n) + P(Y_t > 
    -n)\}\label{eq:PX+PY}
\end{align}
where it suffices to take the infimum over the natural number set only.


\subsection{Probability Bound}
\label{upper_bound}

We next upper bound $P(Y_t>-n)$ using properties of the adversarial mining process, the jumper process, as well as Kroese's results on the difference of renewal processes.  It should not be surprising that the bounds are essentially exponential in the confirmation time when evaluated numerically.

We define 
\begin{align} \label{eq:T}
    T = \min\{t \in [s,s+\Delta]: P_t = P_{s+\Delta}\},
\end{align}
which represents the last time a pacer is mined by $s+\Delta$.
Evidently, 
$T$ is a renewal point, conditioned on which the subsequent pacer arrival times are 
independent of all mining times before $T$.
In particular, $P_{T,T+r}$ is identically distributed as $P'_{0,r}$ for all $r\ge0$.
Also,
the lead $L_s$ and the pacer process $P_{T,d-\epsilon-\Delta}$ are independent conditioned on $T$.  Moreover, $A_{s,d}$ is independent of $T$, $L_s$, and the pacer process.
We evaluate $P(Y_t\le-n)$ by conditioning on and then averaging over $T$:
\begin{align}
    &P(Y_t \leq -n) \notag\\
    &=
    \mathbb{E}[ P(Y_t \leq -n|T) ] \\
    &=
    \mathbb{E}[ P(\sup_{d \geq s+t}\{L_{s} + A_{s,d} -P_{T,d-\epsilon-\Delta} \} \leq -n|T) ] \\
    &=
    \mathbb{E}[ P(\sup_{d \geq s+t}\{L_{s} + A_{s,d} -P'_{0,d-\epsilon-\Delta-T} \} \leq -n|T) ] .
    \label{eq:ELAP}
\end{align}

By~\eqref{eq:T},
$    d-\epsilon-\Delta-T \geq d-s-\epsilon-2\Delta$,
so~\eqref{eq:ELAP} implies
\begin{align}
   & P(Y_t\le-n) \notag \\
     &\geq
     \mathbb{E} [ P(\sup_{d \geq s+t}\{L_{s} + A_{s,d} -P'_{0,d-s-\epsilon-2\Delta} \} \leq -n|T) ] \label{eq:replace_P}\\
     &=P(\sup_{d \geq s+t}\{L_{s} + A_{s,d} -P'_{0,d-s-\epsilon-2\Delta} \} \leq -n)
     \label{eq:Psupd}
\end{align}
where we drop conditioning on $T$ because all other random variables in~\eqref{eq:replace_P} are independent of $T$.
As a Poisson random variable, $A_{s,d}$ is identically distributed as $A'_{0,d-s}$.
Letting $q = d-s$, we rewrite~\eqref{eq:Psupd} as
\begin{align}
   & P(Y_t\le-n)\notag \\
    &\geq
    P(\sup_{q\geq t}\{L_{s} + A'_{0,q} -P'_{0,q-\epsilon-2\Delta} \} \leq -n)\\
    &=
     P(\sup_{q\geq t}\{L_{s} + A'_{0,t} +A'_{t,q}-P'_{0,t-\epsilon-2\Delta}-P'_{t-\epsilon-2\Delta,q-\epsilon-2\Delta} \} \leq -n).
     \label{eq:AP'}
\end{align}
We 
further replace $A'_{t,q}$ by $A''_{0,q-t}$ in~\eqref{eq:AP'}.  
We also note that $P'_{t-2\Delta-\epsilon,q-2\Delta-\epsilon}$ 
statistically 
dominates 
$P''_{0,q-t}$. 
So, we rewrite~\eqref{eq:AP'} as
\begin{align}
   & \quad P(Y_t\le-n) \notag\\
    &
    \geq
    P \left( \sup_{q\geq t}\{L_{s} + A'_{0,t} +A''_{0,q-t}-P'_{0,t-\epsilon-2\Delta}-P''_{0,q-t} \} \leq -n \right) \\
    &=
    P \left( L_{s} + A'_{0,t} - P'_{0,t-\epsilon-2\Delta}
    + \sup_{q\geq t}\{ A''_{0,q-t}-P''_{0,q-t} \} \leq -n \right) \\
    & =
    P \left( L_{s} + A'_{0,t}-P'_{0,t-\epsilon-2\Delta}+M \leq -n \right )\label{eq:final_lowerbound} \\
    &=
    1 - P(P_{0,t-\epsilon-2\Delta}' \leq L_s + A_{0,t}' + M + n-1)\label{eq:upperbound_Y}
\end{align}
where
\begin{align} \label{eq:M=}
    M = \sup_{r\geq 0}\{A''_{0,r}-P''_{0,r}\}.
\end{align}
As the maximum difference between two renewal processes, $M$ has the pmf~\eqref{eq:e(i)} by~\cite[Theorem 3.4]{Kroese_difference} (see also~\cite[Sec.~6]{li2021close}).
given by~\eqref{eq:e(i)}.

Thus, by~\eqref{eq:PX+PY} and~\eqref{eq:upperbound_Y}, we have 
\begin{align}
\begin{split}
\label{eq:1-X+PL}
    P(F_{s,t}^{\epsilon}) &\leq \inf_{n\in \{1,2,...\}}\{P^0(X\ge n)  +P(P_{0,t-\epsilon-2\Delta}' \leq A_{0,t}'+M+L_s+n-1)\}.    
\end{split}
\end{align}
Taking the limit of $\epsilon\to0$, we use continuity to obtain
\begin{align}
\begin{split}
    \lim_{\epsilon \rightarrow 0} P(F_{s,t}^{\epsilon}) &\leq \inf_{n\in \{1,2,...\}}\{P^0(X\ge n) +P(P_{0,t-2\Delta}' \leq A_{0,t}'+M+L_s+n-1)\}\label{eq:1-PX+P}.    
\end{split}
\end{align}
The second term on the right-hand side of~\eqref{eq:1-PX+P} is essentially the probability that the number of pacers mined during $(0,t-2\Delta]$ is fewer than the number of blocks mined by the adversary over a duration of $t$, plus the adversarial advantage after t, plus the adversary's pre-mining gain, plus the number of balanced heights.  The analysis thus far allows us to express this in terms of mutually independent variables $P'_{0,t-2\Delta}$, $A'_{0,t}$, $M$, and $L_s$.  We analyze their individual distributions in the remainder of this subsection.

Let
$L$ denote an independent and identically distributed random variable as $M$.
In Appendix~\ref{a:lead}, we show that $L_s$ is stochastically dominated by $L+A'_{t,t+2\Delta}$.  Hence~\eqref{eq:1-PX+P} still holds with $L_s$ replaced by $L+A'_{t,t+2\Delta}$ to yield
\begin{align}
\begin{split}
    \lim_{\epsilon \to 0}
    P(F_{s,t}^{\epsilon}) \leq&
    \inf_{n\in \{1,2,...\}}\{P^0(X\ge n) + P(P_{0,t-2\Delta}' \leq A_{0,t+2\Delta}'+M+L+n-1)\}
    \label{eq:upperbound_SV}    
\end{split}
\end{align}
where $P'$, $A'$, $M$, and $L$ are independent.

We now invoke the preceding results to prove Theorem~\ref{th:upper_time}.
Basically, we evaluate the upper bound of the probability of violation using 
the distributions (or their bounds) developed for $X$, $M$, and $L$.
Recall the inter-pacer arrival times are $\Delta$ plus exponential variables.  We have in general
\begin{align}
    P( P'_{0,t} \le i-1 )
    &= P( i\Delta + X_1 + \cdots + X_i > t )
\end{align}
where $X_1,X_2,\dots$ are i.i.d.\ exponential with mean $1/h$.  The sum $X_1+\cdots+X_i$ has the Erlang distribution of shape $i$ and rate $h$.  Hence
\begin{align}
    P( P'_{0,t} \le i-1 )
    = 1 - F_2(t-i\Delta; i, h)
\end{align}
where $F_2(\cdot;i,h)$ denotes the Erlang cdf with shape parameter $i$ and rate $h$.

By 
independence 
of those random variables, we have
\begin{align}
    \quad &P(P_{0,t-2\Delta}' \leq A_{0,t+2\Delta}'+M+L+n-1)\notag\\  =&\sum_{i=0}^{\infty}\sum_{j=0}^{\infty}\sum_{k=0}^{\infty}P(A_{0,t+2\Delta}'=i)P(M=j)P(L=k)\cdot P(P_{0,t-2\Delta}' \leq i+j+k+n-1)\\    =&\sum_{i=0}^{\infty}\sum_{j=0}^{\infty}\sum_{k=0}^{\infty}
    f_1(i;\mya (t+2\Delta))e(j)e(k)\cdot (1-F_2(t-(i+j+k+n+2)\Delta;i+j+k+n,\myh))
    \label{eq:infinite_upperbound},
\end{align}
where $f_1(\cdot;\omega)$ denotes the pmf of a Poisson distribution with mean $\omega$.  (The cdf of $P'_{0,t}$ has been evaluated in \cite[Sec.~6.1]{li2021close}.)  It is easy to see that~\eqref{eq:infinite_upperbound} is equal to $G(n,t)$ defined in~\eqref{eq:Gnt}.

By~\eqref{eq:upperbound_SV}, 
the safety of the said transaction confirmed using latency $t$ can be violated with a probability no higher than 
\begin{align}
    \inf_{n \in \{1,2,\dots\}}\{P^0(X\ge n)+G(n,t)\},
\end{align}
which is equal to~\eqref{eq:upper}. 

Since it is numerically inconvenient to compute infinite sums, we further derive an upper bound using only finite sums, which is (slightly) looser than~\eqref{eq:infinite_upperbound}.
For all positive integers $i_0,j_0,k_0$, it is easy to verify that
\begin{align}
    &\quad P(P_{0,t-2\Delta}' \leq A_{0,t+2\Delta}'+M+L+n-1) \notag \\  
    &\leq
    1-P(A_{0,t+2\Delta}'\leq i_0)P(M\leq j_0)P(L\leq k_0)  + P(P_{0,t-2\Delta}' \leq A_{0,t+2\Delta}'+M+L+n,A_{0,t+2\Delta}'\leq i_0,  M\leq j_0,L\leq k_0)\\
    &\leq   \overline{G}(n,t),
    \label{eq:upper_finite}
\end{align}
where $\overline{G}(n,t)$ is calculated as~\eqref{eq:Gnt_}.


Finally, we have that the safety of a given transaction confirmed using latency $t$ can be violated with a probability no higher than 
\begin{align}
\begin{split}
    &\inf_{n \in \{1,2,\dots\}}\{P^0(X\ge n)+\overline{G}(n,t)\}  =
    \inf_{n \in \{1,2,\dots\}}\{1-F(n-1)+\overline{G}(n,t)\}.    
\end{split}
\end{align}

\section{Proof of Theorem \ref{th:lower_time}}
\label{a:lower_time}
In this section, we lower bound the probability of safety violation 
achieved by a specific attack.  The attack is the same as described in Appendix~\ref{a:lower_depth}.
The first new pacer is the first H-block mined after $s$.
Thus the probability that the adversary wins is lower bounded by
\begin{align} \label{eq:LAP'}
    P & \left( 
    \bigcup_{d \geq s+t}\left\{L_{s} + A_{s,d}  \geq 
P'_{d-s-\Delta} \ge 1
\right\} 
 \right) .
 \end{align}
It is not difficult to see that~\eqref{eq:LAP'} is further lower bounded by
\begin{align}
\begin{split}
\label{eq:PU-PU}
P\left( \bigcup_{d \geq s+t}\left\{L_{s} + A_{s,d} \geq 
P'_{d-s-\Delta}
\right\} \right) -
    P \left( \bigcup_{d \geq s+t} \left \{P'_{d-s-\Delta} =0\right \} \right).
\end{split}
\end{align}

The second term in~\eqref{eq:PU-PU} is easily determined as
\begin{align}
    P(  \bigcup_{d \geq s+t} \left \{P'_{d-s-\Delta} =0\right \})
   & =
     P( P'_{t-\Delta} =0)\\
  & =
  e^{-h 
  (t-\Delta)} .
\end{align}

It remains to bound the first term in~\eqref{eq:PU-PU}.
Since $L_{s}$, $A_{s,d}$ and $P'_{d-s-\Delta}$
are independent, we can replace $A_{s,d}$ with identically distributed and independent $A'_{0,d-s}$ and replace $P'_{d-s-\Delta}$
with i.i.d. $P''_{0,d-s}$ in~\eqref{eq:LAP'} to obtain
\begin{align}
    P
    & (\bigcup_{d \geq s+t}\left\{L_{s} + A_{s,d}  \geq P'_{d-s-\Delta}\right\})\notag\\
    &= P(\bigcup_{d \geq s+t}\left\{ L_{s} + A'_{0,d-s}  \geq P''_{0,d-s}\right\})\label{eq:change_1}\\
  &=
    P(\bigcup_{q \geq t}\left\{L_{s} + A'_{0,t}+A'_{t,q}\geq P''_{0,t+\Delta}+P''_{t+\Delta,q}\right\}) . \label{eq:change_2}
\end{align}
Evidently, $A'_{t,q}$ is identically distributed as $A''_{0,q-t}$.  Moreover, as a delayed renewal processes with identical inter-arrival times, $P''_{t+\Delta,q}$ has statistically a larger delay than 
$P'_{\Delta,q-t}$.  Thus $P''_{t+\Delta,q}$ is stochastically dominated by 
$P''_{\Delta,q-t} = P''_{0,q-t}$.  Hence~\eqref{eq:change_2} implies that
\begin{align}
    P
    & (\bigcup_{d \geq s+t}\left\{L_{s} + A_{s,d}  \geq P'_{d-s-\Delta}\right\})\notag\\
    &\geq  P(\bigcup_{q \geq t}\left\{L_{s} + A'_{0,t}+A''_{0,q-t}\geq P'_{0,t+\Delta}+P''_{0,q-t}\right\})\label{eq:change_3}\\
   & \geq P(P'_{0,t+\Delta}\leq  A'_{0,t}+M+L_{s})
\end{align}
where $M$ is defined in~\eqref{eq:M=}.
The pmf of $M$ equals~\eqref{eq:e(i)}.
Similar to \cite{li2021close}, $P(P'_{0, t+\Delta}\leq  A'_{0,t}+M+L_{s})$ can be calculated using the pmfs of the relevant independent random variables.
Thus, given a latency $t$, the safety of a given transaction can be violated with a probability 
no less 
than given in~\eqref{eq:lower}.


\bibliographystyle{ieeetr}
\bibliography{ref}

\end{document}